\documentclass[sigconf, nonacm]{acmart}

\newif\ifrevisionmode
\revisionmodefalse

\newcommand{\publishedvenue}{Proc. VLDB Endow.}
\newcommand{\publishedvolume}{19}
\newcommand{\publishedissue}{9}
\newcommand{\publishedpages}{2168--2182}
\newcommand{\publishedyear}{2026}
\newcommand{\publisheddoi}{10.14778/3819518.3819542}

\usepackage{graphicx} % Required for inserting images
\usepackage{pgfplots}
\usepackage{pgfplotstable} % For reading and processing table data
\usepackage{tikz} % For drawing graphics
\usetikzlibrary{calc} % Coordinate calculations like ($(a)+(b)$)
\usetikzlibrary{arrows.meta,positioning,calc,fit,shapes.misc,shapes.geometric,matrix}
\usetikzlibrary{external}
\makeatletter
\ifnum\pdfshellescape=1\relax
  \immediate\write18{mkdir -p build/tikz-cache build/build/tikz-cache}
\fi
\makeatother
\makeatletter
\let\TPHT@origpgfplotslegendfromname\pgfplotslegendfromname
\renewcommand{\pgfplotslegendfromname}[1]{%
    \tikzexternaldisable
    \TPHT@origpgfplotslegendfromname{#1}%
    \tikzexternalenable
}
\makeatother
\usepackage{filecontents}
\usepackage{enumitem} % For customizing itemize and enumerate environments
\usepackage{subcaption} % For subfigures and subcaptions
\usepackage{amsmath} % For math environments
\usepackage{multirow} % For multi-row table cells
\usepackage{pdflscape} % For landscape pages
\usepackage{booktabs} % For better table formatting
\usepackage{xspace}
\usepackage{wrapfig} % For wrapfigure environment
\pgfplotsset{compat=1.18}
\pgfplotsset{
  every axis/.append style={
    xticklabel style={yshift=2pt},
    yticklabel style={xshift=2pt},
    xlabel style={yshift=5pt},
    ylabel style={yshift=-5pt}
  }
}
\usepgfplotslibrary{statistics} % For statistical plots
\usepgfplotslibrary{fillbetween}
\usepackage{algorithm}
\usepackage[noend]{algpseudocode}
\usepackage[T1]{fontenc}
\usepackage{xcolor} % For text colors
\usepackage[normalem]{ulem} % For \sout strikeout that handles math expressions
\usepackage{soul}
\usepackage{wrapfig}

\usepackage{booktabs}
\usepackage{threeparttable}
\usepackage{caption}

\usepackage{etoolbox}

\newcommand{\tightsize}{\setlength{\tabcolsep}{6pt}\renewcommand{\arraystretch}{1.15}}

\microtypesetup{expansion=true}

\definecolor{exp1}{RGB}{31, 119, 180}    % blue - Cuckoo 
\definecolor{exp2}{RGB}{214, 39, 40}     % red - Iceberg
\definecolor{exp3}{RGB}{127, 127, 127}   % grey - Junction
\definecolor{exp4}{RGB}{44, 160, 44}     % green - TPHT
\definecolor{exp5}{RGB}{148, 103, 189}   % purple - TBB
\definecolor{exp6}{RGB}{255, 127, 14}     % yellow - Blast
\definecolor{exp7}{RGB}{102, 194, 165}    % teal - backup1
\definecolor{exp8}{RGB}{231, 138, 195}    % pink - backup2
\definecolor{poisson}{RGB}{148, 103, 189} % purple - for theoretical plots

\pgfplotsset{
    CuckooBarStyle/.style={fill=exp1!35, draw=exp1!80},
    IcebergBarStyle/.style={fill=exp2!35, draw=exp2!80},
    JunctionBarStyle/.style={fill=exp3!35, draw=exp3!80},
    TPHTBarStyle/.style={fill=exp4!35, draw=exp4!80},
    TBBBarStyle/.style={fill=exp5!35, draw=exp5!80},
    BlastBarStyle/.style={fill=exp6!35, draw=exp6!80},
    htthreeBarStyle/.style={CuckooBarStyle},
    htfourBarStyle/.style={IcebergBarStyle},
    htfiveBarStyle/.style={JunctionBarStyle},
    htoneBarStyle/.style={TPHTBarStyle},
    htsixBarStyle/.style={TBBBarStyle},
    httwoBarStyle/.style={BlastBarStyle}
} 

\pgfplotsset{
    CuckooLineStyle/.style={color=exp1, mark=o, mark size=2.5pt, thick, smooth},
    IcebergLineStyle/.style={color=exp2, mark=square, mark size=2.5pt, thick, smooth},
    JunctionLineStyle/.style={color=exp3, mark=triangle, mark size=2.5pt, thick, smooth},
    TPHTLineStyle/.style={color=exp4, mark=diamond, mark size=2.5pt, thick, smooth},
    TBBLineStyle/.style={color=exp5, mark=star, mark size=2.5pt, thick, smooth},
    BlastLineStyle/.style={color=exp6, mark=pentagon, mark size=2.5pt, thick, smooth},
    BackupLineStyleOne/.style={color=exp7, mark=otimes, mark size=2.8pt, thick, smooth, mark options={solid}},
    BackupLineStyleTwo/.style={color=exp8, mark=x, mark size=2.8pt, thick, smooth, mark options={solid}},
    htthreeLineStyle/.style={CuckooLineStyle},
    htfourLineStyle/.style={IcebergLineStyle},
    htfiveLineStyle/.style={JunctionLineStyle},
    htoneLineStyle/.style={TPHTLineStyle},
    htsixLineStyle/.style={TBBLineStyle},
    httwoLineStyle/.style={BlastLineStyle}
}

\pgfplotsset{
    PoissonLineStyle/.style={color=poisson, thick, smooth},
    PoissonFillStyle/.style={color=poisson!20, fill=poisson!20},
    OccScatterStyle/.style={only marks, mark=o, mark size=2.2pt, color=exp1, opacity=0.85}
}

\pgfplotsset{
    InsertionOnlyStyle/.style={color=exp1, mark=o, mark size=2.5pt, thick, smooth},
    WithDeletionStyle/.style={color=poisson, mark=square, mark size=2.5pt, thick, smooth},
    PoissonTheoryStyle/.style={color=poisson, thick, smooth},
    ExperimentalBoxStyle/.style={fill=exp1!25, draw=exp1}
}

\pgfplotsset{
    CuckooStyle/.style={CuckooBarStyle},
    IcebergStyle/.style={IcebergBarStyle},
    JunctionStyle/.style={JunctionBarStyle},
    TPHTStyle/.style={TPHTBarStyle},
    TBBStyle/.style={TBBBarStyle},
    BlastStyle/.style={BlastBarStyle},
    htthreeStyle/.style={CuckooBarStyle},
    htfourStyle/.style={IcebergBarStyle},
    htfiveStyle/.style={JunctionBarStyle},
    htoneStyle/.style={TPHTBarStyle},
    htsixStyle/.style={TBBBarStyle},
    httwoStyle/.style={BlastBarStyle}
}
\newcommand{\addDataPlot}[5]{%
\addplot[#1] table[y expr=\thisrow{#3 (ops/s)}/1000000, restrict expr to domain={\thisrow{entry_id}}{#5:#5}, restrict expr to domain={\thisrow{case_id}}{#4:#4}, restrict expr to domain={\thisrow{object_id}}{#2:#2}]{\data};%
}

\newcommand{\addAllStandardPlots}[3]{%
\addDataPlot{CuckooStyle}{6}{#2}{#1}{#3}
\addDataPlot{IcebergStyle}{7}{#2}{#1}{#3}
\addDataPlot{JunctionStyle}{15}{#2}{#1}{#3}
\addDataPlot{TBBStyle}{24}{#2}{#1}{#3}
\addDataPlot{TPHTStyle}{17}{#2}{#1}{#3}
\addDataPlot{BlastStyle}{20}{#2}{#1}{#3}
}

\def\getlabelname#1#2{%
    \ifnum#2=0%
        Load%
    \else%
        \ifnum#1=17 Run A\fi%
        \ifnum#1=18 Run B\fi%
        \ifnum#1=19 Run C\fi%
        \ifnum#1=20 Run A$^-$\fi%
        \ifnum#1=21 Run B$^-$\fi%
        \ifnum#1=22 Run C$^-$\fi%
        \ifnum#1=28 Run X\fi%
    \fi%
}

\newcommand{\generateSubfigure}[4]{%
\begin{subfigure}[b]{0.125\textwidth}
\centering
\begin{tikzpicture}
\begin{axis}[
    width=2.8cm,
    height=3.5cm,
    ylabel={\small Throughput (M/s)},
    ylabel style={at={(ticklabel* cs:1.2)}, anchor=south},
    xlabel={#2},
    ybar,
    bar width=3pt,
    xticklabels={},
    xtick style={font=\small,draw=none},
    axis lines=box,
    tick align=inside,
    scaled ticks=true,
    tick label style={font=\small,/pgf/number format/fixed,/pgf/number format/precision=1},
    ymajorgrids=true,
    yminorgrids=true,
    minor tick num=1,
    max space between ticks=35pt,
    try min ticks=5,
    grid style={gray!30},
    ymin=0,
    legend entries = {\htthree, \htfour, \htfive, \htsix, \htone, \httwo},
    legend cell align = left,
    legend style={draw=none, font=\small, legend columns=6, /tikz/every even column/.append style={column sep=0.5cm}},
    legend to name={throughput-legend-horizontal}
]
\addAllStandardPlots{#1}{#3}{#4}
\end{axis}
\end{tikzpicture}
\end{subfigure}%
}

\newcommand{\generateSubfigureNoYLabel}[4]{%
\begin{subfigure}[b]{0.125\textwidth}
\centering
\begin{tikzpicture}
\begin{axis}[
    width=2.8cm,
    height=3.5cm,
    xlabel={#2},
    ybar,
    bar width=3pt,
    xticklabels={},
    xtick style={font=\small,draw=none},
    axis lines=box,
    tick align=inside,
    scaled ticks=true,
    tick label style={font=\small,/pgf/number format/fixed,/pgf/number format/precision=1},
    ymajorgrids=true,
    yminorgrids=true,
    minor tick num=1,
    max space between ticks=35pt,
    try min ticks=5,
    grid style={gray!30},
    ymin=0
]
\addAllStandardPlots{#1}{#3}{#4}
\end{axis}
\end{tikzpicture}
\end{subfigure}%
}

\def\ShiftA{-4.0*\pgfplotbarwidth}
\def\ShiftB{-2.4*\pgfplotbarwidth}
\def\ShiftC{-0.8*\pgfplotbarwidth}
\def\ShiftD{ 0.8*\pgfplotbarwidth}
\def\ShiftE{ 2.4*\pgfplotbarwidth}
\def\ShiftF{ 4.0*\pgfplotbarwidth}

\newcommand{\addResizingPointPlot}[6]{%
  \begingroup
    \edef\YCOL{#4 (ops/s)}%
    \addplot[#1, bar shift=#6] table[
      x expr=#5,
      y expr=\thisrow{\YCOL}/1e6,
      restrict expr to domain={\thisrow{entry_id}}{1:1},
      restrict expr to domain={\thisrow{object_id}}{#2:#2},
      restrict expr to domain={\thisrow{case_id}}{#3:#3}
    ]{\data};
  \endgroup
}

\newcommand{\addOneXGroup}[3]{% #1=metric, #2=case_id, #3=xindex
  \addResizingPointPlot{CuckooBarStyle}{6}{#2}{#1}{#3}{\ShiftA}%
  \addResizingPointPlot{IcebergBarStyle}{7}{#2}{#1}{#3}{\ShiftB}%
  \addResizingPointPlot{JunctionBarStyle}{15}{#2}{#1}{#3}{\ShiftC}%
  \addResizingPointPlot{TBBBarStyle}{24}{#2}{#1}{#3}{\ShiftD}%
  \addResizingPointPlot{TPHTBarStyle}{18}{#2}{#1}{#3}{\ShiftE}%
  \addResizingPointPlot{BlastBarStyle}{21}{#2}{#1}{#3}{\ShiftF}%
}

\newcommand{\addResizingPlots}{%
  \addOneXGroup{fill_throughput}{17}{0}% x=0
  \addOneXGroup{run_throughput}{17}{1}% x=1
  \addOneXGroup{run_throughput}{18}{2}% x=2
  \addOneXGroup{run_throughput}{19}{3}% x=3
  \addOneXGroup{run_throughput}{22}{4}% x=4
}

\newcommand{\readdatafiledir}{./plots/}
\pgfplotstableread[col sep=comma, sci]{\readdatafiledir csv/ycsb_results.csv}{\data}
\pgfplotstableread[col sep=comma, sci]{\readdatafiledir csv/occupancy_results.csv}{\occupancydata}
\pgfplotstableread[col sep=comma, sci]{\readdatafiledir csv/throughput_space_eff_results.csv}{\throughputdata}
\pgfplotstableread[col sep=comma, sci]{\readdatafiledir csv/throughput_space_eff_compact_results.csv}{\throughputcompactdata}
\pgfplotstableread[col sep=comma, sci]{\readdatafiledir csv/progressive_resizing_results.csv}{\progressiveresdata}
\pgfplotstableread[col sep=comma, sci]{\readdatafiledir csv/progressive_resizing_sliding_window_results.csv}{\progressiveresslidingwindowdata}
\pgfplotstableread[col sep=comma, sci]{\readdatafiledir csv/load_factor_support_results.csv}{\loadfactordata}
\pgfplotstableread[col sep=comma, sci]{\readdatafiledir csv/scaling_results.csv}{\scalingdata}
\pgfplotstableread[col sep=comma, sci]{\readdatafiledir csv/small_table_results.csv}{\smalltabledata}
\pgfplotstableread[col sep=comma, sci]{\readdatafiledir csv/tinypointers_comparison_results.csv}{\tpcompdata}
\pgfplotstableread[col sep=comma, sci]{\readdatafiledir csv/occupancy_experimental_box_random.csv}{\occboxrandom}
\pgfplotstableread[col sep=comma, sci]{\readdatafiledir csv/occupancy_experimental_box_sequential.csv}{\occboxsequential}
\pgfplotstableread[col sep=comma, sci]{\readdatafiledir csv/occupancy_experimental_box_low_hamming.csv}{\occboxlowhamming}
\pgfplotstableread[col sep=comma, sci]{\readdatafiledir csv/occupancy_experimental_box_high_hamming.csv}{\occboxhighhamming}
\pgfplotstableread[col sep=comma, sci]{\readdatafiledir csv/occupancy_poisson.csv}{\occpois}
\pgfplotstableread[col sep=comma, sci]{\readdatafiledir csv/data_size_scaling_results.csv}{\datasizedata}

\pgfplotstableread[col sep=comma, sci]{\readdatafiledir csv/ycsb_intro_summary.csv}{\ycsbintrodata}
\pgfplotstableread[col sep=comma, sci]{\readdatafiledir csv/max_space_efficiency.csv}{\maxspaceeffdata}

\pgfplotstableread[col sep=comma, sci]{\readdatafiledir csv/percentile_results.csv}{\percentiledata}

\pgfplotstableread[col sep=comma, sci]{\readdatafiledir csv/query_percentile_results.csv}{\querypercentiledata}

\pgfplotstableread[col sep=comma, sci]{\readdatafiledir csv/mem_growth.csv}{\memgrowthdata}

\pgfplotstableread[col sep=comma, sci]{\readdatafiledir csv/resizing_rss.csv}{\resizingrssdata}

\usepackage{amsthm}            % provides \newtheorem

\theoremstyle{plain}
\newtheorem{theorem}{Theorem}          % global numbering
\newtheorem{lemma}[theorem]{Lemma}     % share the same counter
\newtheorem{claim}[theorem]{Claim}     % share the same counter
\newtheorem{takeaway}{Takeaway}
\theoremstyle{definition}

\newcommand{\qpre}{q^{\text{pre}}}

\definecolor{alexcolor}{HTML}{800080}
\definecolor{xilincolor}{HTML}{BA0202}
\definecolor{billcolor}{HTML}{1E9626}
\definecolor{yuqicolor}{HTML}{8F8C59}
\ifrevisionmode
  
  \DeclareRobustCommand{\alex}[1]{\textcolor{alexcolor}{Alex: #1}}
  \DeclareRobustCommand{\bill}[1]{\textcolor{billcolor}{Bill: #1}}
  \DeclareRobustCommand{\xilin}[1]{\textcolor{xilincolor}{Xilin: #1}}
  \DeclareRobustCommand{\yuqi}[1]{\textcolor{yuqicolor}{Yuqi: #1}}
  
\else
  
  \DeclareRobustCommand{\alex}[1]{}
  \DeclareRobustCommand{\bill}[1]{}
  \DeclareRobustCommand{\xilin}[1]{}
  \DeclareRobustCommand{\yuqi}[1]{}
  
\fi
\soulregister\alex7
\soulregister\xilin7

\newcommand{\htabbr}{TPHT\xspace}
\newcommand{\htone}{Chained-TPHT\xspace}
\newcommand{\httwo}{Flattened-TPHT\xspace}
\newcommand{\htthree}{Cuckoo\xspace}
\newcommand{\htfour}{IcebergHT\xspace}
\newcommand{\htfive}{Junction\xspace}
\newcommand{\htsix}{TBB\xspace}
\newcommand{\htb}{Bucket\xspace}
\newcommand{\htg}{Group\xspace}
\newcommand{\htcp}{Cleary$^{plain}$\xspace}
\newcommand{\htcs}{Cleary$^{sparse}$\xspace}
\newcommand{\htlp}{Layered$^{plain}$\xspace}
\newcommand{\htls}{Layered$^{sparse}$\xspace}

\newcommand{\projurl}{\url{https://github.com/Xilinion/TinyPtr}}
\newcommand{\tphturl}{\url{https://github.com/Xilinion/TPHT}}

\newcommand{\defn}[1]{\textit{\textbf{#1}}}

\ifrevisionmode
  \newcommand{\TODO}[1]{\textcolor{red}{\texttt{TODO:\ #1}}}
\else
  \newcommand{\TODO}[1]{}
\fi

\newcommand{\httwoavgfillthroughput}{226.9}

\newcommand{\httwoavgrunthroughput}{310.4}

\newcommand{\htoneavgfillthroughput}{99.5}

\newcommand{\htoneavgrunthroughput}{153}

\newcommand{\htthreeavgfillthroughput}{116.7}

\newcommand{\htthreeavgrunthroughput}{126.3}

\newcommand{\httwofillspeedup}{1.84}

\newcommand{\htonefillspeeduppercent}{80.9}

\newcommand{\httworunspeedup}{1.94}

\newcommand{\htonerunspeeduppercent}{95.7}

\newcommand{\httwospeeduppercent}{89.3}

\newcommand{\htfouroverhtthreerunthroughputpercent}{26.6}

\newcommand{\htfiveruncthroughput}{293.9}

\newcommand{\htonemaxspaceefficiencypercent}{105.4}

\newcommand{\httwomaxspaceefficiencypercent}{83.4}

\newcommand{\htonememshaveupperpercent}{69.3}

\newcommand{\htonememshavelowerpercent}{22.5}

\newcommand{\htonememshavemeanpercent}{38.6}

\newcommand{\httwomemshavemeanpercent}{22.3}

\newcommand{\htthreethroughputdroppercent}{24.8}

\newcommand{\htfourthroughputdroppercent}{3.9}

\newcommand{\htfivethroughputdroppercent}{48.3}

\newcommand{\htonethroughputdroppercent}{35.3}

\newcommand{\httwothroughputdroppercent}{19.8}

\newcommand{\htfiveoverhttwolowloadratio}{1.4}

\newcommand{\httwoinsertionspeeddifffivezeropct}{30.4}

\newcommand{\httwopositivequerypercentoffastestfivezeropct}{88.5}

\newcommand{\httwonegativequeryspeeddifffivezeropct}{24.1}

\newcommand{\compactbucketingmaxspaceefficiencypercent}{110.3}

\newcommand{\htcpmaxspaceefficiencypercent}{114.7}

\newcommand{\htlpmaxspaceefficiencypercent}{90.1}

\newcommand{\compactbucketinghtonespeedup}{9.8}

\newcommand{\htonebyhtcppercent}{62.5}

\newcommand{\httwooverhtlpspeedup}{2.44}

\newcommand{\htlsperfdropvshtlppercent}{64.1}

\newcommand{\htoneovercompactminspeedup}{1.37}

\newcommand{\htoneovercompactmaxspeedup}{20.13}

\newcommand{\htlsspacegainvshtlppercent}{8.5}

\newcommand{\insertiononlyloadfactorpercent}{98.2}

\newcommand{\deletionincludedloadfactorpercent}{95}

\newcommand{\htonethreadscalingfactor}{0.908}

\newcommand{\httwothreadscalingfactor}{0.925}

\newcommand{\httworesizingavgrunthroughput}{245}

\newcommand{\htfourresizingthroughputdecreasepercent}{35.9}

\newcommand{\httworesizingfillthroughputdecreasepercent}{76.2}

\newcommand{\httworesizingthroughputdroppercent}{27.4}

\newcommand{\htfourworstresizingspaceefficiencypercent}{32.6}

\newcommand{\htsevenworstresizingspaceefficiencypercent}{46.7}

\newcommand{\tpoursspeedupmin}{2.17}

\newcommand{\tpoursspeedupmax}{3.69}

\newcommand{\htonesmallnegquerythroughput}{75.3}

\newcommand{\htonescalingsmallavgthroughput}{36.9}

\newcommand{\httwoscalingsmallavgthroughput}{43}

\newcommand{\htonescalingsmallavgspeedupoverbest}{1.02}

\newcommand{\httwoscalingsmallavgspeedupoverbest}{1.19}

\newcommand{\httwoscalinglargeavgthroughput}{19.1}

\newcommand{\httwoscalinglargeavgspeedupoverbest}{1.11}

\ifrevisionmode
  \DeclareRobustCommand{\rev}[1]{\textcolor{blue}{#1}}
  \newcommand{\mymarginpar}[1]{\marginpar{\color{red}\scriptsize #1}}
\else
  \DeclareRobustCommand{\rev}[1]{#1}
  \newcommand{\mymarginpar}[1]{}
\fi

\begin{document}
\title{Succinct and Fast Tiny Pointer Hash Tables}

\author{Xilin Tang}
\affiliation{%
  \institution{Cornell University}
  \city{New York}
  \state{NY}
  \country{USA}
}
\email{xilin@cs.cornell.edu}

\author{Yuqi Mai}
\affiliation{%
  \institution{Cornell University}
  \city{New York}
  \state{NY}
  \country{USA}
}
\email{ym562@cornell.edu}

\author{William Kuszmaul}
\affiliation{%
  \institution{Carnegie Mellon University}
  \city{Pittsburgh}
  \state{PA}
  \country{USA}
}
\email{kuszmaul@cmu.edu}

\author{Alex Conway}
\affiliation{%
  \institution{Cornell Tech}
  \city{New York}
  \state{NY}
  \country{USA}
}
\email{me@ajhconway.com}

\begin{abstract}

Hash tables sit on the critical path of many systems, yet modern designs still force a trade-off between fast operations and high memory overhead.
We revisit this trade-off and present Tiny Pointer Hash Tables (TPHT), a family of practical hash tables that make two ideas from theory work at system scale: compressing pointers down to a byte, and encoding keys compactly so less metadata is needed.
We engineer these ideas into two complementary designs.
Chained-TPHT targets maximal space savings, \rev{ and is to the best of our knowledge the first simple and practical succinct hash table design, achieving a footprint less than the total data size with constant-time operations.}
Flattened-TPHT targets latency, organizing data to keep the common case within a single cache miss while retaining strong space efficiency.
Both variants support dynamic resizing without global pauses and integrate cleanly with 64-bit keys and values.

Across YCSB and microbenchmarks, TPHT advances the latency-space Pareto frontier: Chained-TPHT reaches \htonemaxspaceefficiencypercent\% space efficiency, and Flattened-TPHT achieves \httwomaxspaceefficiencypercent\% space efficiency with up to \httwospeeduppercent\% higher throughput than strong baselines.
Together, these results show that techniques primarily known in theory can be turned into production-ready hash tables that meaningfully reduce memory use while delivering state-of-the-art performance.

\end{abstract}

\maketitle
\begin{center}
\small
This is an author-prepared arXiv version of a paper published in
\publishedvenue\ \publishedvolume(\publishedissue):\publishedpages,
\publishedyear.
The version of record is available at
\href{https://doi.org/\publisheddoi}{\publisheddoi}.
Source code and experimental artifacts are available at \projurl;
a standalone implementation for direct use is available at \tphturl.
\end{center}

\section{Introduction}\label{sec:introduction}

Hash tables underpin a vast range of data-intensive systems, from distributed databases~\cite{10.14778/3025111.3025113, abebe2020morphosys, alquraan2020scalable} and in-memory key-value stores~\cite{dewitt1984implementation, garcia2002main, wang2022case, 196284} to network monitoring~\cite{snoeren2001hash, 10.1145/3716819}, genomics analysis~\cite{todd2016parallel, muller2017metacache, wood2019improved}, and real-time analytics engines~\cite{ashraf2015understanding, hua2014fast}.
Because hash tables promise expected \(O(1)\) queries, they are often the default choice for high-throughput applications.
However, the size of datasets has exploded and continues to grow, while the cost of DRAM limits the practical memory capacity of servers.
Nonetheless, state-of-the-art low-latency hash tables use only a fraction of their memory for data, with substantial portions used by metadata and collision resolution~\cite{swisstables,junction,10.1145/3588727}.

%The key issue is that collision resolution typically requires substantial memory overhead.
%For example, with separate chaining every pointer consumes a substantial amount of memory, and the resulting indirection causes additional caches misses leading to higher expected and tail latencies.
%On the other hand, open-addressing schemes---linear probing, cuckoo hashing, Robin Hood, and their variants---do not use pointers and so do not incur these memory overheads.
%However, the performance of these systems degrades sharply as occupancy rises, and so they require low load factors to stay fast in practice~\cite{baeldung2024hashmap}.
%Recent designs such as Google's Swiss Tables~\cite{swisstables} and Meta's F14 Hash Table~\cite{folly} mask this slowdown with bucketization and per-entry metadata to enable SIMD intrinsics.
%However, this added metadata increases memory overhead, and as open-addressing schemes they still face load-factor tradeoffs. 
%Across these approaches, pointer overhead, sparse layouts, and auxiliary bits routinely waste tens of bytes per key.

This paper revisits the fundamentals of hash tables and poses two questions:
\begin{itemize}
    \item \textbf{Space.} How close can a practical hash table drive memory usage to the information-theoretic lower bound?
    \item \textbf{Speed.} Can it still rival the throughput of today's state-of-the-art designs while consuming markedly less memory?
\end{itemize}

We show how to build fast space-efficient hash tables using \defn{tiny pointers}~\cite{tinypointersSODA} and \defn{quotienting}~\cite{10.5555/280635,10.5555/548089,Arbitman2009BackyardCH}, theoretical techniques to compress pointers and keys respectively.
We first demonstrate how these ideas from theory can be practically implemented, and then use them in two different hash-table designs.
The first replaces the pointers in a chaining hash table with tiny pointers, which together with quotienting, results in an ultra space-efficient hash table; the table can even be smaller than the underlying data!\footnote{Minimal and dynamic perfect hashing~\cite{Belazzougui2009HashDA,doi:10.1137/S0097539791194094} can approach even smaller footprints by compressing hash functions for fixed or slowly changing key sets, but they typically target membership or static dictionaries rather than storing full key-value pairs with general updates.}
The second uses tiny pointers and quotienting in an open-addressing hash table to pack keys and collision-resolution metadata efficiently into each cache line, which yields a space-efficient hash table with state-of-the-art performance.

\paragraph{Tiny Pointers}
Tiny pointers are implemented through a data structure called a \defn{dereference table}.
A tiny pointer compresses a full-word address to just \(\min(\log \delta^{-1}, \log \log \log n)\)~bits, where \(1-\delta\) is the load factor of the dereference table.
Tiny-pointer theory uses recursive multi-level dereference tables to achieve theoretically optimal space usage, but this is difficult to implement in practice.
\mymarginpar{R2D6}\mymarginpar{R4W1}
One contribution of this paper is to offer a simpler implementation based on power-of-two-choice scheme, which empirically \rev{works without failure at a load factor of} \(\insertiononlyloadfactorpercent\%\).

In our hash tables tiny pointers use 8~bits instead of the 64~bits of a standard machine word or the \(O(\log n)\)~bits that a straightforward index would require.

\paragraph{Quotienting}
The core idea behind quotienting is easy to understand if we consider a chained hash table that stores fingerprints (i.e.\ hashed keys) instead of the keys themselves.
The chain that a key belongs to is determined by the high-order bits of its fingerprint, referred to as a \defn{quotient}.
Therefore, these bits are already implicitly determined by each item's location, and so the table need only store the low-order bits of the fingerprint, the \defn{remainder}.
This technique is used in quotient filters~\cite{quotientfiltersVLDB} to build a near optimal filter data structure.

However, hash tables need to store the key itself, not just the fingerprint, because fingerprints can collide.
\rev{To solve this problem, we use a technique, one-round Feistel permutations, from the theory literature to encode the}\mymarginpar{R2D1} keys themselves so that they can be quotiented without losing the original key.
\rev{This allows} us to build real-world hash tables that use quotienting to save roughly \(\log n\)~bits of space per key.

\paragraph{Tiny-Pointer Hash Tables}
We use tiny pointers and quotienting to design two novel high-performance space-efficient hash tables.

\paragraph{\htone}
Our first design emphasizes space efficiency by using a chained hash table design, with a head array storing (tiny) pointers to the first link in each chain.
The links in the chains use tiny pointers and the keys are stored using quotienting.
Perhaps surprisingly, under most configurations, \htone simultaneously achieves lower memory usage than \emph{even storing the underlying data in plain-text form} and strong overall performance.

\htone uses tiny pointers to reduce the space consumed by pointers in the head array and the chain links.
Quotienting then recovers enough memory from the keys to fully compensate for this (much reduced) overhead.

Because the head array uses tiny pointers, its size is a small fraction of the data size.
This means that it can have sufficiently many entries that each chain is short, so that operations incur few cache misses.
Ideally, in a system with a 64\,MB~L3 cache, the head array for a 1\,GB~table can fit entirely in cache.

Despite being one of the oldest and simplest data structures in computer science\footnote{First introduced in 1953 by Hans Peter-Luhn, chained hash tables predate both unbalanced and balanced binary trees. In fact, according to Knuth, the chained hash table may be the first ever use of a \emph{linked list} in a computer program~\cite{10.5555/280635}}, chained hash tables have since been considered to be impractical in space-sensitive settings because of the many pointers that the data structure requires.
\htone turns this conventional wisdom on its head, demonstrating that, with the right algorithmic techniques, chained hashing can actually be implemented as an \emph{extremely} space-efficient data structure.

\paragraph{\httwo}
While \htone achieves excellent memory-efficiency, there is still a latency overhead due to the indirection of having to first access the head array and then access the data itself.
Our second design, \httwo, is less memory-efficient, but allows most operations to be completed using a single cache miss, leading to better average latency.
Furthermore, operations average 1.25~cache misses and never exceed three (outside resizing), which yield a strong tail latency profile.

\httwo accomplishes this by replacing the head array with an array of cache-line-sized home groups.
Each key hashes to a home group, which itself stores 0--4 key-value pairs.
When more than 4 key-value pairs hash to the same home group, some ``overflow'' key-value pairs are stored in a dereference table, and are referenced from the head group using tiny pointers.
The home entries are accessed in a single cache miss, whereas the overflow entries require two cache misses.
The memory-efficiency tradeoff occurs because the head groups will not have uniform occupancy, so some memory in head groups will go unused.

Consistent with prior work~\cite{10.1145/3588727,10.14778/3389133.3389134,10.1145/3309206,dramhit,katsarakis2024dlht}, both variants support 64-bit keys and values, deletions, and online resizing.
As in previous hash-table designs, they can be straightforwardly extended to variable-length keys and values by storing pointers to the actual keys and values in the hash table.

\textbf{Results.}
\htone attains \htonemaxspaceefficiencypercent\% space efficiency, reducing memory by \htonememshavemeanpercent\% relative to state-of-the-art baselines.
\httwo delivers \httwomaxspaceefficiencypercent\% space efficiency and cuts memory by \httwomemshavemeanpercent\% while running \httwospeeduppercent\% faster than the same baselines.

\textbf{Contributions.}
This paper contributes the following:
\begin{itemize}
    \item \textbf{Practical tiny pointers.}
        We present a simple and practical implemention of  tiny pointers, which empirically results in excellent compression with efficient dereferencing.
    \item \textbf{Practical key quotienting.}
        We incorporate quotienting into real-world hash tables, saving \(\log n\)~bits per key.
    \item \textbf{Beyond 100\% space efficiency.}
        Using these two techniques, \htone \rev{is \htoneovercompactminspeedup-\htoneovercompactmaxspeedup$\times$ the speed of prior hash tables whose footprint is} below the raw data size.\mymarginpar{R2D1}
    \item \textbf{State-of-the-art performance with less space.}
        \httwo pairs the same techniques with a cache- and SIMD-aware layout to surpass state-of-the-art tables in performance while using less memory.
%    \item \textbf{In-flight cooperative resizing.}
%        A lightweight framework lets worker threads assist an ongoing resize without global pauses or transient memory spikes, providing smooth concurrency as the table grows.
\end{itemize}

\rev{The paper also comes with a contribution of theoretical interest. From a theory perspective, \htone is a \emph{succinct hash table}, a hash table that uses space within a factor of $(1 + o(1))$ of the information-theoretic optimum, while offering $O(1)$ expected-time operations.\footnote{In the standard parameter regime where keys are of size $w = (1 + \Theta(1))\log n$ bits, \htone uses space $B + O(n \log \log n) = (1 + O(\log \log n / \log n))B$ bits, where $B$ is the information-theoretic optimum.} Such hash tables have existed in the theory literature (e.g., \cite{raman2003succinct,Arbitman2009BackyardCH,10.1145/3519935.3519969,bender2024modern}), but are considered too complicated (and with too poor of constants) to be practical. (As we will discuss in Section~\ref{sec:background}, there are practical \emph{compact hash tables} \cite{10.1007/s00453-022-00996-y,2019kopple,10.1109/TC.1984.1676499,poyias2017mbonsaipracticalcompactdynamic}, which can achieve $(1 + o(1))$ of the space optimum but not with $O(1)$ operations.) \htone, in addition to being practical, is arguably also the simplest succinct hash-table design to date.}

\paragraph{Artifacts and standalone implementation}
The research prototype and experimental artifacts are available at \projurl.
We also provide a standalone implementation of \htabbr intended for direct use by practitioners at \tphturl.

%\xilin{Need rewriting after other parts are finished.}
%\xilin{Do we need this overview at all?}
%We begin by revisiting the theoretical lower bounds on space usage for hash tables within the word RAM model, which establishes the context and motivation for our work in Section~\ref{sec:space_bound}.
%The core of our approach is the concept of "tiny pointers," whose functionality and potential for space saving we explore in detail in Section~\ref{sec:tinyptrlib}.
%Building on this, we show how to combine tiny pointers with quotienting to construct a hash table that achieves a sub-data-size footprint in Section~\ref{sec:htone}.
%We then introduce a flattened data layout that optimizes memory access patterns, enabling performance competitive with state-of-the-art designs in Section~\ref{sec:httwo}.
%Finally, we complete our design by presenting robust mechanisms for concurrency control and dynamic resizing, ensuring its practicality in real-world systems in Section~\ref{sec:conc_resize}.

\section{Background}\label{sec:background}

% In this section we discuss several common hash table designs, emphasizing how these designs lead to memory-performance tradeoffs and other memory overheads.

Most hash table designs can be broadly be divided into two categories, \defn{separate chaining} and \defn{open addressing}.

\paragraph{Separate chaining}
A separate chaining hash table hashes each key to a bucket, which is then stored as a linked list.
Each list is referred to as a \defn{chain}.
The canonical design uses a \defn{head array}, each entry of which stores a pointer to the first link in its chain.
Variations on this design optimize for resizing and performance bottlenecks. Split-Ordered Lists~\cite{10.1145/1147954.1147958} enable lock-free incremental resizing while preserving bucket semantics, and CLHT~\cite{10.1145/2786763.2694359} improves cache locality by packing multiple entries into cache-line-sized mini-buckets.

A major source of memory overhead in separate chaining are the pointers in the chains.
Each item must store a pointer, so for example if the hash table stores 64-bit keys and values and uses 64-bit pointers, this immediately leads to a 50\% memory overhead.
However, additionally, assuring performance in the table often requires even more memory overheads.

Separate chaining performance is dominated by cache misses incurred during operations: queries must access the head pointer and traverse the chain until we find the key.
Therefore, for performance reasons, it's important to have enough buckets that the chains are short.
However, doing so introduces more null pointers in the head array, causing additional memory overhead.
Some designs store key-value pairs in the entries in the head array.
This removes the indirection but further inflates entry size and leaves many entries empty.\footnote{With \(n\) elements and \(n\) buckets, about \(1/e \approx 0.37\) of buckets are empty.}
In CLHT, it actually packs each link as a cache-line mini-bucket (up to three 64-bit tuples inside).
This improves the locality of chain traversals, since multiple items can be retrieved in a single cache line.
However, the imperfect load balancing typically leaves over 62\%\cite{10.1145/3588727} of memory unused.

So in terms of memory efficiency, separate chaining starts with a large overhead from pointers and then must decide how much additional memory overhead to trade for performance reasons.

\paragraph{Open addressing}
An open addressing hash tables store all entries directly in a single array, using a probe sequence to resolve collisions.
The most common variant, linear probing, searches sequentially from the hash location until finding an empty slot or the target key~\cite{flajolet1998analysis,braverman2024tight}.
Other schemes like quadratic probing~\cite{kuszmaul2024towards}, double hashing~\cite{guibas1976analysis}, and cuckoo hashing~\cite{pagh2004cuckoo} use different probe sequences to distribute keys and reduce clustering effects.

The primary advantages of open addressing are data locality and implementation simplicity.
However, performance degrades rapidly as the load factor increases due to longer probe sequences and increased cache misses.
To maintain good performance, most open addressing schemes operate at load factors, typically below 75\%~\cite{baeldung2024hashmap}, leaving substantial memory unused.
%\mymarginpar{R2D1}
%\xilin{Even for compact schemes~\cite{10.1109/TC.1984.1676499,poyias2017mbonsaipracticalcompactdynamic}, the load factor is still a problem as discussed in Section~\ref{subsec:space_bound}.}

Recent designs attempt to mitigate these issues through metadata optimizations.
Google's Swiss Tables~\cite{swisstables} and Meta's F14 Hash Table~\cite{folly} use metadata bytes to quickly skip over occupied slots during probing.
%IcebergHT~\cite{10.1145/3588727} employs similar per-entry tags to reduce probe distances, but these metadata bits add memory overhead.
%Ordered linear probing (also known as Robin Hood hashing) minimizes variance in probe distances by displacing existing entries, improving worst-case performance but requiring additional bookkeeping.

%Taken together, there appear to be fundamental tradeoffs between time and space, preventing the design of hash tables that are both space and time efficient simultaneously.
%Separate chaining pays pointer and head-array overheads to bound indirection, whereas open addressing eliminates pointers but requires relatively low load factors and per-entry metadata to keep probes short.
%%Thus, practical hash table designs spend nontrivial memory to sustain performance.
%The main contribution of this paper is to show how these tradeoffs can be substantially improved by using tiny pointers and quotienting.

\mymarginpar{R2D1}
\paragraph{Compact Hash Tables}
\rev{Recall that a hash table is \emph{compact} if it uses space within a factor of $1 + \delta$ of optimal, for some $\delta \in (0, 1)$, while supporting operation times that are a function of $\delta^{-1}$. This is a weaker guarantee than being \emph{succinct}, which means that the hash table can support $\delta = o(1)$, with $O(1)$-time operations. Whereas (until this work) succinct hash tables have been theoretical only, several designs of compact hash tables have been shown to be practically feasible. (See, also, Section~\ref{subsec:space_bound} and Section~\ref{subsec:compact_ht_comparison} for comparisons with \htone.)}

\rev{Like separate chaining, compact bucketing~\cite{10.1007/s00453-022-00996-y,2019kopple} uses variably sized buckets to store entries, but instead of storing them as linked lists, they store them on the heap and reallocate them as necessary.
This allows the hash table to be space efficient, but the cost of reallocation significantly hurts update performance and causes memory fragmentation.}

\rev{Compact linear-probing schemes~\cite{10.1109/TC.1984.1676499} use quotienting together with metadata bits to store partial key information, but they suffer a performance penalty at high load factors due to long probe sequences.
Layered versions~\cite{poyias2017mbonsaipracticalcompactdynamic} use multiple layers to reduce probe distances and allow for better performance at high load factors.}

\section{Practical Tiny Pointers}\label{sec:tinyptrlib}

Tiny pointers offer a path to compressing pointers in data structures.
In this section we describe the user interface to tiny pointers and how to integrate them into \rev{a chained hash table example}.

\subsection{Tiny Pointers and Dereference Tables}\label{subsec:tinyptrderef}

\paragraph{The tiny-pointer user interface}
\rev{Tiny pointers are made possible through a data structure called a \defn{dereference table}, which functions as a specialized memory allocator that will contain the objects (in our case, nodes in the hash table) that the tiny pointers will point at.
Given a \defn{capacity} $n$ and a \defn{object size} $s$, we create a dereference table of size roughly $n\cdot s$ bytes, capable of storing up to $n$ objects of $s$ bytes each.
    The dereference table provides the following user interface for tiny pointers:
    \mymarginpar{R1D2}
    \begin{enumerate}
        \item \texttt{Allocate}$(k)$ allocates $s$ bytes in the dereference table, and returns a tiny pointer $p$.
              The parameter $k$ is the ID for the object, and must be unique among all other IDs that are currently active.
              The same ID $k$ must be provided in order to dereference or free $p$ in the future.
        \item \texttt{Dereference}$(k, p)$ returns a (non-tiny) pointer to the object with ID $k$.
              The prerequisite is that the ID $k$ is active with tiny pointer $p$.
        \item \texttt{Free}$(k, p)$ deallocates the object with ID $k$.
              The prerequisite is that $k$ is active with tiny pointer $p$.
    \end{enumerate}
}

\paragraph{Differences between traditional and tiny pointers. }
Tiny pointers go hand-in-hand with the \defn{dereference table}, which functions as a specialized memory allocator.
%The dereference table partitions its roughly $n\cdot s$ bytes amongst the up-to-$n$ allocations that exist at any time.
Whereas standard memory allocators allocate objects of different sizes and are primarily concerned with fragmentation, dereference tables allocate objects of fixed size $s$ and don't face this issue.
\rev{However, this means that they cannot be used to allocate variable-sized objects without using an additional level of indirection.}
\mymarginpar{R1D4}
On the other hand, objects in a dereference table must be able to be referenced via a 1-byte tiny pointer, whereas standard allocators can use 8-byte pointers.
That this is possible is non-obvious and requires significant theoretical techniques~\cite{tinypointersSODA}.

Notably, even though each object in the dereference table is given a user-assigned ID $k$, these IDs are not actually \emph{stored} in the dereference table.
(Indeed, the entire dereference table consists of roughly $n \cdot s$ bytes, which is only enough space to store the allocated $n$ objects.)
Rather, it is the \emph{user}'s job to provide the user-assigned ID $k$ whenever a tiny pointer is dereferenced or freed.

Finally, tiny pointers \emph{preserve} a key property of traditional pointers: following a pointer requires a single memory access.
As we will see later on in our proposed implementation, the translation of tiny pointer to pointer (via a call to \texttt{Dereference}$(k, p)$) incurs no additional memory accesses.

% \paragraph{The bounty: compressed pointers for a smaller and faster data structure}
% For a chained hash table, we can replace each 64-bit pointer with an 8-bit tiny pointer.
% This can also help improve performance as the head array is 8 times smaller and is likely to fit in a higher level cache.

\subsection{Practical Implementation}\label{subsec:tpimpl}

\paragraph{Allocating tiny pointers with two-choice bucketing}
Recall that a dereference table is initialized a maximum capacity $n$ and a object size $s$.
The dereference table will consist of $N\cdot s$-byte slots, where $N$ is selected to be slightly larger than $n$ (say, $n \approx 0.95 N$).
As we shall see, the choice to have $N$ slightly larger than $n$ (even just, say, 5\% larger) will be critical for good pointer compression, as it allows flexibility in which slot to use for each allocation.

The dereference table consists of two arrays:
\begin{itemize}
    \item \textbf{Meta table} ($2n$ bytes).
          Each of the $n$ \defn{bin}s stores a two-byte header: the current free-slot count and the offset of the first free slot\rev{, i.e., the list head of the free list}.\mymarginpar{R2D5}
    \item \textbf{Data table} ($\approx n$ entries).
          Each consecutive block of $(2^{7}-1)$ entries form a bin; free slots inside a bin are threaded into a free list via 7-bit offsets.
\end{itemize}

\begin{figure}[h]
    \centering
    \includegraphics[width=0.9\linewidth]{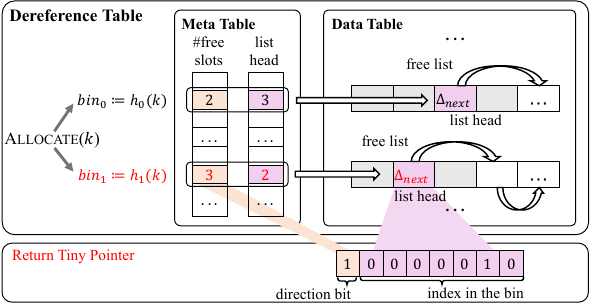}
    \caption{Meta layout, data bins, and the steps of \texttt{Allocate}$(k)$: two hashes, choose emptier bin, pop head, return 8-bit $p$.}\label{fig:tpimpl}
\end{figure}

\noindent To perform \texttt{Allocate}$(k)$, we:
\begin{enumerate}
    \item Look at the free-slot counters for bins $h_0(k)$ and $h_1(k)$, and let $h_i(k)$, $i \in \{0, 1\}$, be the emptier one.
    \item Pop the head of the free-slot list in bin $h_i(k)$ to obtain a free slot $j \in \{1, \ldots, 127\}$, and decrement the bin's free-slot counter.
          This step fails if the bin is fully occupied---we will return to this shortly.
    \item Return the 8-bit tiny pointer consisting of the \defn{direction bit} $i$ concatenated with the slot index $j$.
\end{enumerate}

This procedure is illustrated in figure~\ref{fig:tpimpl}.
Note that we reserve the tiny pointer $0$ to use as a null pointer---consequently we use bins of size $2^7-1$, rather than $2^7$.

\noindent Similarly, to perform \texttt{Dereference}$(k, p)$, we:
\begin{enumerate}
    \item Interpret $p$ as a direction bit $i$ and a slot index $j$.
    \item Return a pointer to the $j$-th slot of bin $h_i(k)$.
\end{enumerate}

Note that \texttt{Dereference}$(k, p)$ is just a computation and does not require a memory access.

\noindent Finally, to perform \texttt{Free}$(k, p)$, we:
\begin{enumerate}
    \item Interpret $p$ as a direction bit $i$ and a slot index $j$.
    \item Increment the free-slot counter for bin $h_i(k)$, and add $j$ to the bin's free list.
\end{enumerate}

\paragraph{The main challenge: avoiding allocation failures. }
In this basic design, an allocation may fail if both of the bins $h_0(k)$ and $h_1(k)$ are fully occupied.
The key design principle behind a dereference table is to allocate tiny pointers in such a way that these failures almost never happen.

To minimize allocation failures, our proposed design relies heavily on a probabilistic phenomenon known as the \defn{power of two choices}~\cite{10.1137/S009753970444435X,DBLP:journals/corr/TalwarW13,doi:10.1137/S0097539795288490,bansal2022balancedallocationsheavilyloaded}.
This says that, if $a$ balls are placed into $b$ bins, and if each ball is placed in the emptier of two \emph{random} bins, then (with very high probability) all of the bins will have very similar loads (no bin will have more than $a/b + \log \log b + O(1)$ balls).
Unlike many theoretical phenomena, the power-of-two choices has very good real-world constants.
\rev{Importantly, allocation failure probability is not simply a per-operation constant; it guarantees that no failure occurs w.h.p. for polynomially many operations.}\mymarginpar{R4W1}\mymarginpar{R4D3}\mymarginpar{R4D4}

Our implementation reaches a sustained load factor well beyond the \deletionincludedloadfactorpercent\%, suggested by analyses in Section~\ref{subsec:load_factor}.

\paragraph{Difference between our implementation and the theory}
In~\cite{tinypointersSODA}, dereference tables are designed as a complex structure with multiple hash-table-like levels.
The bits of each tiny pointer are used to both identify the hash functions and the index in the substructure of hash tables.
This design involves complex multi-level load-balancing as well as careful partitioning of the tables, which makes it difficult and inefficient to implement in practice.
It further uses exhaustive tabulation (also known as the method of Four Russians) to implement some operations in constant time, which does not work well in practice\footnote{\rev{Our implementation's average throughput is \tpoursspeedupmin$\times$--\tpoursspeedupmax$\times$ that of a prior third-party implementation~\cite{tinypointers_github} directly following the theoretical work~\cite{tinypointersSODA}.}}.\mymarginpar{R2D4}
% \xilin{Experimental results suggest that the average throughput of our design is from \tpoursspeedupmin$\times$ to \tpoursspeedupmax$\times$ of the throughput of simply following prior theoretical work.}
% \alex{My problem with saying this here is that we're not including those results in the paper, so I'm not totally comfortable making this claim. I think we can either omit it or leave it as a footnote where we cite the github and say this.}
% \cred{Moved to footnote.}
% \alex{Could add something here about the differences between our design and the github one, or just the performance differences.}

% \paragraph{A simple example: chained hash tables}
\subsection{Example: Chained Hash Tables}\label{subsec:chained_ht_example}

We now demonstrate how tiny pointers can be used in the context of a simple chained hash table.
%To \xilin{instantiate} tiny pointers in the context of hash tables, we use a running example: a chained hash table.
This example will take us part of the way towards understanding \htone, which will go further by integrating quotienting.
% We assume that our hash table stores 64b keys and 64b values and holds up to $n$ key-value pairs.

% In our hash table example, $s$ should be the size of a node (128 bits + the size of a pointer), and $n$ should be the hash table's capacity.
%In Section~\ref{subsec:resizing}, we'll discuss how to handle resizing.
% (In this basic construction, tiny pointers must know $n$ and $s$ up front---in Section~\ref{subsec:resizing}, we'll discuss how to handle resizing.)

Having initialized the dereference table, we can now allocate nodes for the hash table.
In a chained hash table, each node lies in a linked list and is referenced by a unique tiny pointer.
\rev{To allocate a node, we call \texttt{Allocate} with a unique ID for that node, and store the returned tiny pointer in the appropriate place (e.g., the head array for the bucket, or the next pointer of another node).}
\mymarginpar{R1D2}

%To allocate the node corresponding to a pointer with address $A$, we call \texttt{Allocate}$(A)$, which allocates an $s$-byte object in the dereference table and returns a 1-byte tiny pointer $p$ to store at $A$.
%So for example, if $A$ is the head-pointer address for a bucket, it holds the tiny pointer $p$ for the first node $u$ in its chain.
%Address $A$ plays a crucial role here: \texttt{Allocate} requires an \defn{ID} for that tiny pointer, which is a user-specified name for the object being allocated.
%This name must be unique across objects in the dereference table, and is required whenever operating on the tiny pointer, e.g., \texttt{Dereference}$(A,p)$ to obtain a 64-bit pointer to $u$ or \texttt{Free}$(A,p)$ to release $u$.

One might be tempted to use the key $k$ stored in the node as the unique ID.
This works when $k$ is the first node in the chain.
However, if another key $\ell$ is stored in the first node, the head tiny pointer $p$ will be allocated with ID $\ell$, so calling \texttt{Dereference}$(k, p)$ will not return the correct node (or even necessarily a valid node).
Therefore it is important that the IDs for the tiny pointers in each chain be computable from the bin index and the chain prefix.
Moreover, the ID must be unique across objects in the dereference table, and is required whenever operating on the tiny pointer, e.g., \texttt{Dereference}$(k,p)$ to obtain a 64-bit pointer to $u$ or \texttt{Free}$(k,p)$ to release $u$.
This highlights an important general challenge to using tiny pointers: determining unique and computable IDs.

\mymarginpar{R1D2}
\rev{
    In this hash-table example, we solve this challenge by using the address $A$ of each tiny pointer $p$ as its unique ID $k$.
    This means that the two bins where the node pointed to by $p$ can be allocated are $h_0(A)$ and $h_1(A)$, where $h_0$ and $h_1$ are the two hash functions used in our dereference table implementation.
    Notice that this ID satisfies both the uniqueness and computability requirements: each tiny pointer has a unique address, and the address is known whenever the tiny pointer is used.
    See Figure~\ref{fig:tp_ht_together} for an example query.
}
\begin{figure}[h]
    \centering
    \includegraphics[width=1\linewidth]{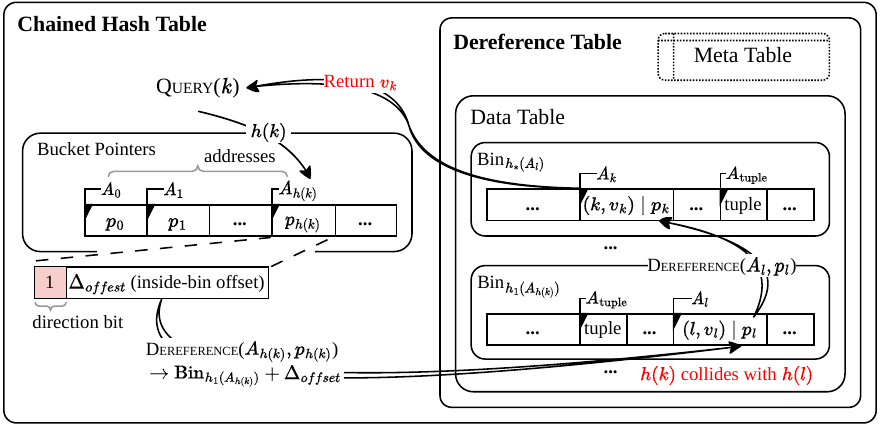}
    \caption{\rev{Query for key $k$ in the chained hash table with tiny pointers. The direction bit of the first tiny pointer is 1 in this example, and leads to the bin corresponding to $h_1(A_{h(k)})$ in our dereference table implementation. The meta table is not used during queries.}}\label{fig:tp_ht_together}
\end{figure}
% \mymarginpar{R1D2}

%\xilin{For example, in the chained hash table example, the owner ID $k$ used for these operations is exactly the address of the tiny pointer, and the tiny pointer $p$ then points to the node.}

\section{Quotienting}\label{sec:quotienting}

The idea behind quotienting is to use the bin index to implicitly store part of each key.
If the keys themselves are random or if it's sufficient to store fingerprints rather than the original keys, then the high-order bits of the random key or fingerprint can be directly used to determine the bin index.
Then these high-order bits do not need to be stored themselves and can be reconstructed from the bin index when needed.

The use of quotienting in hash tables was suggested by Knuth~\cite{10.5555/280635}, and has been applied in quotient filters~\cite{quotientfiltersVLDB,quotienthashtablesSIGAPP,10.1145/3448016.3452841,10.14778/3523210.3523211}, which can be interpreted as quotienting hash tables which store fingerprints rather than keys.
There is a literature of hash tables that use quotienting in theory to achieve memory-efficiency~\cite{10.1145/3625817,Arbitman2009BackyardCH,quotienthashtablesSIGAPP}\rev{, and some of these have practical implementations~\cite{2019kopple,10.1109/TC.1984.1676499,poyias2017mbonsaipracticalcompactdynamic,10.1007/s00453-022-00996-y,compactparallelhashtablesongpu}}.
\mymarginpar{R2D1}

The difficulty with using quotienting is that keys which share a bin don't have to share a prefix.
One possible way to handle this would be to use an invertible hash function---a random permutation---to replace each key with a fingerprint in a way that preserves the in the original key.
Such an approach could, in principle, allow quotienting to be used as described above.
The problem is that, in general, no known techniques to generate random permutations are viable in practice, and in fact, there aren't even techniques to do so in theory~\cite{broder1998min}.

To get around this, we use a construction called a one-round Feistel permutation~\cite{10.5555/548089}.
Rather than trying to emulate a fully random permutation hash function, this construction makes a specified bit range of the key ``look random'' and then uses that range for quotienting.
Feistel permutations have been used in theoretical hash-table designs~\cite{Arbitman2009BackyardCH}\rev{, as well as in a GPU implementation~\cite{compactparallelhashtablesongpu}}.
\mymarginpar{R2D1}

\subsection{One-Round Feistel Permutations}\label{subsec:hashfunc}

For a key $k$, set $k = \qpre \circ r$, where $\qpre$ is the prefix of $k$ that we will quotient off, $r$ is the remainder suffix.
Let $w$ be the word length and $2^{w-\ell}$ be any suitable divisor.
Starting from any hash family $\mathcal{H}$ of functions $h:[2^w]\!\to\![2^{\ell}]$, we define $\mathcal{H}_q=\{h' \mid h'(\qpre \circ r)=h(r)\oplus \qpre , h\in\mathcal{H}\}$, where $\oplus$ denotes XOR.
Intuitively, $h'$ infuses randomness from $h(r)$ into the prefix of the key.
\mymarginpar{R1D6}
\rev{Importantly, if two keys have the same $h'$ value, they can be compared using only the stored remainder, and each key can be reconstructed from the remainder and $h'$.}
We realize the one-round Feistel permutation when calculating $h(r) \oplus \qpre$\rev{, which we refer to as the \defn{quotient}}.
\rev{Algorithm~\ref{alg:quotienting} shows this quotienting process.}
% We now establish the properties of the induced family $\mathcal{H}_q$ used in our design.
\mymarginpar{R1D2}

\rev{
\begin{algorithm}[h]
  \caption{\textsc{One-Round Feistel Quotient}}\label{alg:quotienting}
  \small
  \begin{algorithmic}[1]
    \Require Key $k=\qpre \circ r$, hash function $h$
    \Procedure{\texttt{Quotient}}{$k$}
    \Return $q \gets h(r) \oplus \qpre$
    \EndProcedure
    % \Statex
    \Procedure{\texttt{RecoverKey}}{$q, r$}
    \Return $k \gets (q \oplus h(r)) \circ r$
    \EndProcedure
  \end{algorithmic}
\end{algorithm}
}

Note that $\mathcal{H}_q$ need not preserve the independence properties of $\mathcal{H}$.
In particular, even if $\mathcal{H}$ is fully random, $k_1 = \qpre_1 \circ r$ and $k_2 = \qpre_2 \circ r$ will never collide if $\qpre_1 \neq \qpre_2$.
Therefore $\mathcal{H}_q$ is not even pairwise independent.
Following~\cite{doi:10.1137/0217022}, we show that $\mathcal{H}_q$ is universal and has strong distributional properties---in particular, Chernoff bounds hold and so its max load is the same as $\mathcal{H}$'s.

\begin{claim}
  $\mathcal{H}_q$ is universal when $\mathcal{H}$ is pairwise independent.
\end{claim}

\begin{proof}
  For distinct keys $k=\qpre_1 \circ r_1$ and $k'=\qpre_2 \circ r_2$, pairwise independence of $\mathcal{H}$ implies
  \[
    \Pr_{h'\sim\mathcal{H}_q}
    \bigl[h'(k)=h'(k')\bigr]
    =\Pr_{h\sim\mathcal{H}}\bigl[h(r_1)=h(r_2)\oplus \qpre_1 \oplus \qpre_2\bigr]
    \le 2^{-\ell}.
  \]
\end{proof}

% We now show that Chernoff-type concentration holds for $\mathcal{H}_q$, when $\mathcal{H}$ is mutually independent.
% Intuitively, outcomes that are guaranteed to differ (e.g., $0\circ0$ vs. $1\circ0$ produce $h(0)$ and $h(0)\oplus1$) cannot collide, but the hash family behaves otherwise like an independent one.

\begin{claim}
  If $\mathcal{H}$ is mutually independent, then $\mathcal{H}_q$ has expected max load $\Theta(\log{n}/\log\log{n})$.
\end{claim}

\begin{proof}
  The indicators $X^{(b)}_i=\mathbf{1}[h'(k_i)=b]$ (with $h'\sim\mathcal{H}_q$) are negatively associated, so standard Chernoff bounds apply to $X^{(b)}=\sum_i X^{(b)}_i$\cite{Wajc2017NegativeA}.

  Thus, by the standard result~\cite{10.5555/646975.711521}, the Chernoff bounds imply that the expected max load is $\Theta(\log{n}/\log\log{n})$.
\end{proof}

In practice we instantiate $\mathcal{H}$ with \textsf{xxHash}~\cite{xxhash_2024}.

\section{Chained Tiny Pointer Hash Table}\label{sec:htone}

In this section, we present a chained hash table that fuses \emph{tiny pointers} with \emph{quotienting}.
This design approaches the information-theoretic space lower bound.
The resulting structure is highly space-efficient; as demonstrated in~Section~\ref{subsec:space_bound}, its footprint can even fall below the size of the stored data.

At a high level, \htone is a chained hash table that stores all nodes in a single high-load-factor dereference table.
Further, \htone uses quotienting to reduce the size of stored keys.

\subsection{Design of \htone}\label{subsec:chaining}

This section explains how quotienting and tiny pointers combine to form our chained variant, \htone.

Let $n$ denote the number of stored tuples and let $2^{w-\ell}$ be the divisor chosen for quotienting.
We refer to the set of keys with a given quotient $q$ as $q$'s \defn{quotient group}.
In \htone, chains precisely correspond to quotient groups, so we will sometimes refer to chains and quotient groups interchangeably.
The data structure comprises two tables:
\begin{itemize}
    \item \textbf{Head array} ($2^{\ell}$ entries).
          Each cell holds a tiny pointer to the first node in the chain for the corresponding quotient group (or \texttt{NULL} if the group is empty).
    \item \textbf{Dereference table} ($\approx n$ slots).
          Each slot stores a triple $(p,r,v)$ consisting of the tiny pointer that links to the next element, the quotiented key (remainder), and the value.
          In the ideal case this takes $8-\ell+2w$ bits per tuple; we provision a small slack to stay below the critical load.
\end{itemize}

\paragraph{Chain indentification}
Figure~\ref{fig:quotient_hash} illustrates the data path: for a key $k$ we compute its quotient using the hash function from Section~\ref{subsec:hashfunc}.
The quotient is used to determine which chain $k$ belongs to, and in particular is used as the index in the head array.
When accessing keys in that chain, we undo the quotienting by appending the quotient $q$ to the stored tuple $(r,\textit{value})$.

\begin{figure}[h]
    \centering
    \includegraphics[width=0.9\linewidth]{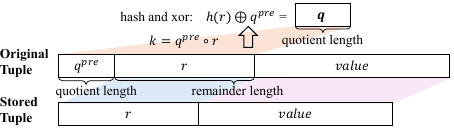}
    \caption{Quotienting pipeline: the hash determines the quotient and only the remainder is stored in the chain.}\label{fig:quotient_hash}
\end{figure}

\paragraph{Monolithic dereference table}
As in a traditional chained hash table, tuples in a chain (quotient group) are organized as a linked list.
All chain nodes across different chains are to be stored in a flat memory pool, which is a single dereference table.
So, even though the chains are logically partitioned, they are physically intermingled in a single table.

\paragraph{Tiny pointer IDs in \htone}
As discussed in Section~\ref{subsec:tinyptrderef}, there is some subtlety in how IDs are associated to tiny pointers in a chained hash table.
In particular, recall that when allocating a node, we cannot use the key $k$ that it stores as the ID.
Instead for the ID we use the address of the unique [tiny] pointer in the list which points to the node.
Importantly these IDs are unique and known at dereference-time (since we obtain the tiny pointer itself from its address).

\paragraph{Hash table operations}
The procedure for inserting a key-value pair $(k,v)$ is detailed in Algorithm~\ref{alg:chained_insert}.
First, the key $k$ is split into a prefix $\qpre$ and a suffix $r$ to compute a bucket index (quotient) $q = h(r) \oplus \qpre$.
Denoting the head array by $B$, we then traverse the linked list starting from $B[q]$, using the in-place address of each tiny pointer's location as its ID for dereferencing.
Upon reaching the end of the chain, a new entry is allocated in the dereference table and linked into the list.
The new entry is then populated with the quotient, the value, and a \texttt{NULL} tiny pointer.

Queries, updates, and deletes follow a similar procedure to locate the target entry by traversing the appropriate chain.
Once found, the operation is completed by interacting with the dereference table via the interfaces from Section~\ref{subsec:tpimpl}.

\newcommand{\Null}{\texttt{NULL}}
\newcommand{\Addr}{\operatorname{addr}}

\begin{algorithm}[h]
    \vspace{1pt}
    \caption{\textsc{Insert}}\label{alg:chained_insert}
    \small
    \begin{algorithmic}[1]
        \Require Global head array $B$, key $k=\qpre \circ r$, value $v$
        \State $q \gets h(r) \oplus \qpre$ \Comment{Quotient group identifier}
        \State $p^{\text{tp}} \gets \Addr(B[q])$ \Comment{Pointer to the tiny pointer}
        \State $k^{\text{deref}} \gets \Addr(B[q])$ \Comment{In-place key for dereference}
        \While{$*p^{\text{tp}} \neq \Null$} \Comment{Follow to the end of the chain}
        \State $e \gets \texttt{Dereference}(k^{\text{deref}}, *p^{\text{tp}})$
        \State $p^{\text{tp}} \gets \Addr(e.p);\enspace k^{\text{deref}} \gets \Addr(e.p)$ \Comment{Advancing}
        \EndWhile
        \State $e,\; *p^{\text{tp}} \gets \texttt{Allocate}(k^{\text{deref}})$ \Comment{Reserve new entry and link it}
        \State $e.k \gets r;\enspace e.v \gets v;\enspace e.p \gets \Null$
    \end{algorithmic}
\end{algorithm}

\subsection{Performance Advantages}

Beyond space efficiency, \htone offers two key performance advantages over conventional chaining methods: enhanced cache locality and shorter average probe chains.

\paragraph{Enhanced cache locality}
A conventional chained hash table relies on a head array of 8-byte pointers to locate the head of each key group's collision chain.
Replacing these with 1-byte tiny pointers shrinks the head array's memory footprint by a factor of eight.
This size reduction greatly improves the chances that the array stays in CPU cache---in the best case, the entire head array can fit in cache.
Consequently, a higher fraction of lookups to find a chain's head become cache hits, reducing latency.

\paragraph{Shorter probe chains}
The average probe chain length is determined by the ratio of elements to key groups.
In traditional designs, expanding the head array to increase the number of key groups incurs a substantial space overhead and impacting cache efficiency.
The minimal footprint of tiny pointers essentially eliminates this trade-off.
With \htone, it's inexpensive to size the head array so as to maintain an element-to-group ratio at or below one by setting $2^{\ell}\ge n$.
This reduces the average chain length to one or less, allowing most lookups to complete in a single probe; we lower-bound this probability by $1-e^{-n/2^{\ell}}\ge 63.2\%$, which we also confirm empirically.
A minor limitation is diminishing returns from further enlarging the head array: while multi-element chains shrink slowly, a larger array can exceed cache and erode locality.

\paragraph{A remark on the symbiosis between tiny pointers and quotienting}
We remark that part of the reason we are able to practically use quotienting in \htone is because chained hash tables, \emph{a priori}, are much easier to perform quotienting on than are open-addressed hash tables (which require carefully managed metadata for quotienting to be possible).
Historically, the compatibility between chaining and quotienting was viewed as essentially pointless, since chaining pays so much space overhead for pointers.
However, with the introduction of tiny pointers, we are able to get the best of both worlds---the low metadata overhead of open addressing, with the convenient quotienting of chaining.

\subsection{Space Analysis}\label{subsec:space_bound}

We now show that \htone is nearly space optimal.

How small \emph{can} a hash table be?
Consider the word-RAM model with a word size of $w$ bits ($w=64$ on most contemporary machines).
Storing $n$ key-value pairs verbatim consumes $2n$ words, or $2nw$ bits.
Yet a hash table represents only the \emph{set} of keys without regard to their order.

Assuming all keys are distinct, the number of possible key sets is $\binom{2^{w}}{n}$.
Any representation therefore requires at least
\[
    \begin{aligned}
        % \log \binom{2^w}{n} & = 2^w\, H_2(p) - \frac{1}{2} \log \bigl(2 \pi n(1-p) \bigr) + o(1) \\
        %                     & \approx n\bigl(w-\log n\bigr) + n\log e + o(n)
        \log \binom{2^w}{n} \approx n\bigl(w-\log n\bigr) + n\log e + o(n)
    \end{aligned}
\]
% bits, where the approximation follows from Stirling's formula and $p=n/2^w$, $H_2(p)=-p\log\frac{1}{p}-(1-p)\log\frac{1}{1-p}$.
bits, where the approximation follows from Stirling's formula.
It holds with approximately 1\% error when $1000n\le 2^{w}$.
Relative to the na\"{\i}ve $nw$-bit allocation, the entropy bound saves roughly $\log n$ bits \emph{per element}.

For \htone, putting the pieces together, besides each $w-bit$ value, each key stores its quotient in $w-\log n$ bits and one 8-bit tiny pointer; the head array contributes $8n$ bits and the meta table $16n/127$ bits. Accounting for load factor $1-\delta$ (empirically $\delta \approx 2\%$), the total footprint is
\[
    \frac{n\bigl(2w-\log n + 8 + 16/127\bigr)}{1-\delta} + 8n \;\text{bits},
\]
which is nearly optimal relative to the entropy bound while preserving constant-time operations. \rev{We remark that, from a theoretical perspective, tiny pointers of $O(\log \log n)$-bits support a dereference table with load factor of $1- \delta = (1 - 1/ \log n)$ \cite{tinypointersSODA}. This means that, in the standard parameter regime of $w = O(\log n)$, \htone uses $B + O(n\log \log n)$ bits, where $B = \log \binom{2^w}{n}$ is the information-theoretically optimal space usage.}

% \mymarginpar{R2D3}
% \xilin{Compared with other compact hash tables, \htone may pay a larger constant-factor $n$-bit overhead for certain baseline design choices, but evaluation (Section~\ref{subsec:compact_ht_comparison}) shows substantially higher performance. Primitive bidirectional linear probing~\cite{10.1109/TC.1984.1676499} uses $n(2w-\log n+5)/(1-\delta)$ bits but the displacement needed to recover the key hash typically impedes performance. Layered designs~\cite{poyias2017mbonsaipracticalcompactdynamic} use $n(2w-\log n+3+\frac{14}{9}+o(1))/(1-\delta)$ bits by storing displacement information across multiple levels. Bucketing~\cite{10.1007/s00453-022-00996-y} uses $n(2w-\log n+O(\log\log w))$ bits but is significantly slower due to frequent bucket reallocations.}

\mymarginpar{R2D3}
\rev{Table~\ref{tab:compact_ht_space_op} summarizes the space and operation complexity for compact hash tables.
Here, instantiation means the practical structure that results from adapting the theoretical design in parameters and structure (e.g., fitting a $\log \log n$ component into a byte).
Notably, at load factor $1-\delta$, the insertion (and sometimes query) cost of prior compact hash tables depends inversely on $\delta$, so at near-optimal space $(1+o(1))B$ they incur $\omega(1)$ or $O(\log n)$ operation.
\htone, by contrast, achieves both $(1+o(1))B$ space and $O(1)$ operations, making it a simple and practical \emph{succinct} hash table.
}

\begin{table}[h]
    \centering
    \caption{\rev{Asymptotic and instantiation space redundancy over the approximate information-theoretic optimum $n(2w - \log n)$ in bits and expected operation complexity for compact hash tables. In this table only, load factor is written as $1/(1+\bar\delta)$ (elsewhere $1-\delta$) to keep the formulas compact. Here $w=O(\log n)$. For \htone, $\bar\delta^*=o(1)$, distinct from the $\bar\delta=\Theta(1)$ used for the other methods.}}\label{tab:compact_ht_space_op}

    \resizebox{\columnwidth}{!}{%
        \tightsize
        \begin{tabular}{lllcc}
            \hline
            \textbf{Method}                                         & \textbf{Space Redundancy} & \textbf{Instantiation Redundancy}                    & \multicolumn{1}{l}{\textbf{Insert}}  & \multicolumn{1}{l}{\textbf{Query}} \\ \hline
            \htone                                                  & $O(n\log \log n)$         & $\bar\delta^*(w + 8 + 16/127)n + 8n$   & \multicolumn{2}{c}{$O(1)$}                                                \\
            Cleary~\cite{10.1109/TC.1984.1676499}                   & $O(\bar\delta n \log n)$      & $\bar\delta(w+5)n$                   & \multicolumn{2}{c}{$O(\bar\delta^{-2})$}                                      \\
            Layered~\cite{poyias2017mbonsaipracticalcompactdynamic} & $O(\bar\delta n \log n)$      & $\bar\delta(w+3+14/9+o(1))n$ & $O(\bar\delta^{-2})$                     & $O(\bar\delta^{-1})$                   \\
            Group~\cite{10.1007/s00453-022-00996-y}                 & $O(n)$                    & $(4+1+1)n$                             & $O(\log n)$                          & $O(1)$                             \\
            Bucket~\cite{10.1007/s00453-022-00996-y}                & $O(n\log \log n)$         & $(8+1)n$                             & \multicolumn{2}{c}{$O(\log n)$}                                           \\ \hline
        \end{tabular}%
    }
\end{table}

\section{Flattened Tiny Pointer Hash Table}\label{sec:httwo}

While \htone prioritizes space efficiency with chaining and tiny pointers, \httwo aims to push the speed-space Pareto frontier.
To do so, \httwo trades some compression headroom for better cache locality and therefore better latency.

Specifically, we focus on three key improvements:
\begin{itemize}
    \item \textbf{Eliminating chain indirection:} We remove the level of indirection required to access chains, thereby reducing memory accesses and latency.
    \item \textbf{Flattening chains:} We replace the linked-list traversal in collision chains with a compact, cache-friendly data layout, which minimizes cache misses and improves locality.
    \item \textbf{Vectorizing key comparison:} We leverage SIMD instructions to parallelize the key comparison process during lookups, significantly accelerating query throughput.
\end{itemize}
These optimizations collectively enable \httwo to achieve performance that is comparable to, and in some cases surpasses, current state-of-the-art hash tables.

\subsection{\httwo Design}\label{subsec:flatten}

To minimize memory access latency, \httwo employs a cache-efficient flattened data layout.
% To minimize memory access latency, \httwo employs a flattened data layout that is designed to maximize cache efficiency.
Specifically the layout is designed with three goals in mind: eliminating head pointer indirection, maximizing cache line locality, and flattening collision chains.

\paragraph{Head pointer elimination}
Accessing a quotient group in \htone via its base pointer array incurs an extra memory indirection.
In \httwo, rather than using a head array, we use a \defn{home array}.
In the home array, entries are \defn{home blocks}, which store up to a fixed number of \defn{home tuples}, as well as pointers to \defn{overflow tuples}.
One can think of this as removing the home tuples from the dereference table and storing them directly in the array.
As in the home array, each home block corresponds to a quotient $q$, so that all keys with quotient $q$ can be found in (as a home tuple) or via (as an overflow tuple) that home block.

\paragraph{Cache line locality}
We carefully structure the home blocks so that they fit in aligned 64-byte cache lines.
This allows us to fit up to 4 tuples directly in the home block.
With sufficiently large quotients, there is still room for metadata and tiny pointers to overflow tuples.
However, all we really need is to quotient at least a single byte from each key.
This means that $\ell-2 \ge 8$, implying a minimum of $2^{10}=1024$ data inputs, a reasonable scale.

Here is a place where \httwo incurs two sources memory overhead.
First, at reasonable load factors, there will be a significant number of underfull home blocks.
Second, in order to fill these home blocks, we want to average around 4 keys per quotient group.
This means that we can quotient off 2 fewer bits than before.

Alternatively, one can use a single home block for multiple quotient groups and thus quotient off more bits. However, this would require a larger metadata overhead to identify the inside quotient groups and also prohibits the following chain flattening optimization to improve tail latency.

\paragraph{Chain flattening}
So far we haven't discussed how to handle overflow tuples.
We could use a chained design as in \htone, but then traversing the chains would require multiple memory hops, which can significantly degrade tail latency.
Instead, we flatten the chain by embedding pointers to all overflow tuples of a quotient group directly within the home block.
This ensures that accessing any element in a collision chain requires at most one additional cache miss.
In cases where a quotient group has more members than fit in the home block, we migrate one resident key-value pair to the dereference table and repurpose its slot to store a small vector of tiny pointers (one per overflow element).
The evicted pair is also referenced by this vector.
This keeps all pointers in the home blocks and avoids pointer chasing.

Note that even when 4 [home] tuples are stored in the home block, there is still room for multiple tiny pointers.
Indeed, suppose that we are quotienting at least 2 bytes per key, which happens when $\ell-2 \ge 16$, i.e. $n \ge 2^{18} = 262144$.
Then with 4 home tuples there are still 8 remaining bytes.
Excluding 2 of which we will use for metadata, we have 6 bytes that can hold 3 fingerprint-tiny-pointer pairs, so there must be more than 7 tuples sharing a quotient before a tuple needs to be migrated.
A Poisson-approximation analysis shows that this migration occurs with probability 5.1\%.

\subsection{Vectorized Lookups}

%To leverage the single-instruction-multiple-data (SIMD) capabilities of modern processors, we enhance our design with vectorized lookups.
%This optimization necessitates a modification to the cache-line layout detailed in Section~\ref{subsec:flatten}.

We now alter the home-block structure to leverage the SIMD capabilities of modern processors.

The core of our vectorized approach is to extract a one-byte fingerprint from each stored tuple.
Each home block stores an array of these fingerprints at the beginning of each line, one corresponding to each tuple, including both home and overflow tuples.

During a lookup, we first perform a vectorized comparison of the query key's fingerprint against all fingerprints in the array using a single SIMD instruction, such as \texttt{\_mm\_cmpeq\_epi8}.
A full key comparison is then performed only for those tuples whose fingerprints match.
This two-phase process effectively prunes the search space within a home block, significantly reducing the number of key comparisons and pointer traversals.

A important design consideration is the generation of these fingerprints.
A naive approach of extracting the fingerprint from the key's remainder (e.g., using its most significant byte) could lead to many fingerprint collisions if many keys have the same high-order bits.
Therefore, we enlarge the quotient length by another byte when applying the Feistel permutation in Section~\ref{subsec:hashfunc} and truncate the most significant byte for the fingerprint.

\subsection{Structure and Operations}\label{subsec:flattened_operations}

The layout of \httwo is shown in Figure~\ref{fig:flatten}, for cases with 4 home tuples and 3 home tuples (after migration due to overflow).

The layout includes the following components:
\begin{itemize}
    \item \textbf{\rev{Home array}}\mymarginpar{R2D5} ($n/4$ cache lines): Array of 64B home blocks. Each home block entry compacts fingerprints, tiny pointers, and quotiented key-value pairs together with metadata consisting of the number of fingerprints and tiny pointers, and concurrency control information.
    \item \textbf{Dereference table} ($\approx n/4$ slots by Section~\ref{subsec:failure-rate}). This unchanged from \htone (Section~\ref{subsec:chaining}), except that each entry is smaller due to the addition quotienting.
\end{itemize}

\begin{figure*}[ht]
    \centering
    \includegraphics[width=0.9\textwidth]{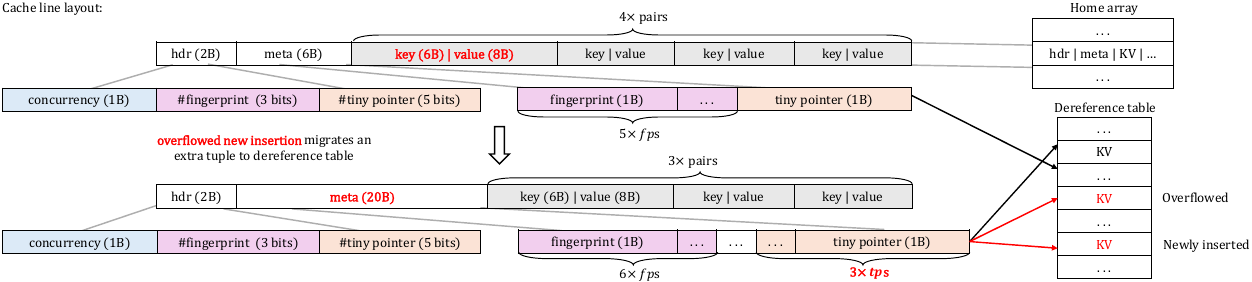}
    \caption{Flattened Data Layout of \httwo}\label{fig:flatten}
\end{figure*}

The procedure for inserting a tuple $(k, v)$ into \httwo is detailed in Algorithm~\ref{alg:flattened_insert}.
The key is hashed to a home block $b$, and then further quotiented to generate a fingerprint, which is stored in the fingerprint array.
The algorithm then proceeds with one of three insertion procedures based on the available space in the home block.
If sufficient space exists, the tuple is stored directly in the home block.
Otherwise, the tuple is added to the dereference table $A$, and a tiny pointer to it is stored in the home block.
If the home block is full, a home tuple is first evicted to the dereference table to make space for a pointer to the new entry.

Query, update, and deletion operations are similar, using vectorized SIMD instructions to scan fingerprints and identify candidate entries before performing a full key comparison.

\begin{algorithm}[h]
    \caption{\textsc{Insert}}\label{alg:flattened_insert}
    \small
    \begin{algorithmic}[1]
        \Require \rev{home array} $C$, key $k=\qpre_0 \circ r_0$, value $v$
        \State $b \gets h(r_0) \oplus \qpre_0$ \Comment{Target home block index}
        \State $\qpre_1 \gets r_0 \div 2^{w-\ell-8}$;\enspace $r_1 \gets r_0 \bmod 2^{w-\ell-8}$
        %  \Comment{Quotient \& remainder for fingerprint}
        \State $fp \gets h(r_1) \oplus \qpre_1$ \Comment{Fingerprint of the key}

        \State $L \gets C[b]$
        %  \Comment{Retrieve the home block}
        \State \texttt{AddFingerprint}$(L, fp)$ \Comment{Store fingerprint in the home block}

        \If{\texttt{HasInCacheSpace}$(L)$} \Comment{Case 1: Space for in-cache KV}
        \State \texttt{AddInCacheKV}$(L, r_1, v)$
        \ElsIf{\texttt{HasPointerSpace}$(L)$} \Comment{Case 2: Space for a tiny pointer}
        \State $e, tp \gets \texttt{Allocate}(L.tp.end());\enspace e.k \gets r_1;\enspace e.v \gets v$
        %  \Comment{Allocate for new entry}
        % \State $e.k \gets r_1;\quad e.v \gets v$
        % \Comment{Fill the new entry}
        \State \texttt{AddTinyPointer}$(L, tp)$
        % \Comment{Store pointer in home block}
        \Else \Comment{Case 3: Evict an in-cache KV to make space}
        \State $r^{\text{evict}}, v^{\text{evict}} \gets \texttt{EvictInCacheKV}(L)$
        % \Comment{Evict KV to free space}
        \State $e_1, tp_1 \gets \texttt{Allocate}(L.tp.end());\enspace e_1.k \gets r^{\text{evict}};\enspace e_1.v \gets v^{\text{evict}}$
        % \Comment{Allocate for evicted KV}
        % \State $e_1.k \gets r^{\text{evict}};\quad e_1.v \gets v^{\text{evict}}$
        \State \texttt{AddTinyPointer}$(L, tp_1)$
        %  \Comment{Store pointer to evicted KV}

        \State $e_2, tp_2 \gets \texttt{Allocate}(L.tp.end());\enspace e_2.k \gets r_1;\enspace e_2.v \gets v$
        %  \Comment{Allocate for new KV}
        % \State $e_2.k \gets r_1;\quad e_2.v \gets v$
        \State \texttt{AddTinyPointer}$(L, tp_2)$
        %  \Comment{Store pointer to new KV}
        \EndIf
    \end{algorithmic}
\end{algorithm}

\subsection{Failure Rate Analysis}\label{subsec:failure-rate}

While this flattened design improves cache performance, it introduces two new failure risks.
First, each home block has a fixed size and must accomodate all the data it stores either as home tuples or as tiny pointers to overflow tuples.
Second, the dereference table must be sized so that it can accomodate all overflow tuples, while maintaining high space utilization.

While bad events such as home block overflows and dereference-table exhaustion are possible, our probabilistic analysis verifies that these events do not occur with high probability across all workloads.
Specifically, we prove that home blocks operate well within their capacity limits w.h.p., and the dereference table can be conservatively sized to handle overflows without risk of exhaustion.

%\paragraph{Problem formulation}
%We model the hash table operation under worst-case conditions where the number of tuples is $n=2^\ell$, resulting in $n/4$ home blocks and a dereference table that stores tuples overflowed from home blocks.
%In our analysis, we assume the most challenging scenario where only one byte is saved from each key during compression, yielding key-value pairs of 14 bytes each (excluding the fingerprint).
%This conservative assumption ensures our bounds hold even under minimal quotienting.

\paragraph{Home block capacity analysis}
In particular, if sufficiently many tuples share a home block, even if they are all overflow tuples, there won't be enough space for their tiny pointers.
We refer to this event as a \defn{hard overflow}.
We now show that the probability of a hard overflow is negligible with a classical balls-and-bins analysis of $n$ keys (balls) into $m = 2^{\ell-2}$ home blocks (bins).

\begin{lemma}[Maximum Bin Size~\cite{10.5555/646975.711521}]
    Throw $n$ balls into $m$ bins i.u.d., where $m/\mathrm{polylog}(m) \le n \ll m \log m$.
    Let $M$ be the max load.
    Then, $M \le \frac{\log m}{\log \log m + \log m - \log n} (1 + o(1))$ w.h.p.
\end{lemma}

Applying this lemma with $n$ keys and $m=2^{\ell-2}$ home blocks, we establish that the max load on any home block is at most $\frac{\ell}{\log \ell - 2} (1 + o(1))$ w.h.p.
For any dataset with fewer than $2^{64}$ items, this is fewer than 24 elements.
Each home block can hold up to 31 tiny pointers (considering
metadata and fingerprints), so this is well within the capacity.
Therefore, hard overflows are statistically negligible.

\paragraph{Dereference table sizing analysis}
Let $O$ denote the number of overflow tuples in an arbitrary home block.
By linearity of expectation, the expected total number of overflow tuples is $(n/4)\mathbb{E}[O]$.

If a home block has $x$ tuples, then $O = \max(0, x - 4 + \lfloor x/6 \rfloor)$.
The R.V.\ $x$ follows a binomial distribution, $x \sim \operatorname{Binomial}(n, 4/n)$, which we approximate by $\operatorname{Poisson(4)}$~\cite{10.5555/1076315}.
Using the fact that, for $x \sim \operatorname{Poisson(4)}$, we have $\Pr[X = i] = 4^i e^{-4} / i!$, we can conclude that
$$
    \mathbb{E}[O] \approx \sum_{i=5}^{\infty} (i - 4 + \lfloor i/6 \rfloor)\frac{4^i e^{-4}}{i!} < 0.997,
$$
which corresponds to an expected overflow ratio below 24.9\%.
This ratio improves with dataset size; when $n\ge 2^{16} = 65536$, this ratio decreases to less than 20.8\%.

We now use the Azuma-Hoeffding inequality to establish concentration bounds around the expected overflow size.

%\begin{theorem}[Azuma-Hoeffding~\cite{Azuma1967WEIGHTEDSO}]
%    Let ${(X_t)}_{t=0}^n$ be a martingale with respect to a filtration $(\mathcal{F}_t)$ and bounded differences: for all $t$, $|X_t - X_{t-1}| \le c$ almost surely. Then for any $\lambda>0$,
%    \[
%        \Pr\bigl[\,X_n - X_0 \ge \lambda\bigr] \le \exp\left( -\, \frac{\lambda^2}{2n c^2} \right).
%    \]
%\end{theorem}

\begin{claim}
    The overflow size of \httwo exhibits strong concentration around its expected value w.h.p.
\end{claim}

\begin{proof}[Proof]
    Let $X_i$ be the home block identifier for the $i$-th tuple, and let $S = f(X_1, \ldots, X_n)= \sum_{i=1}^{n/4} O_i$ denote the total overflow size across all home blocks.
    Consider the Doob martingale $M_t = \mathbb{E}[S \mid X_1,\ldots,X_t]$ for $t=0,\ldots,n$.
    Revealing one tuple can change the total overflow by at most two units, hence the martingale differences are bounded: $|M_t - M_{t-1}| \le 2$ almost surely.
    By Azuma-Hoeffding, for any $\lambda>0$,
    \[
        \Pr\bigl[\,S - \mathbb{E}[S] \ge \lambda\bigr] \le \exp \Bigl( -\, \tfrac{\lambda^2}{2 n \cdot 2^2} \Bigr).
    \]
    Choosing $\lambda = 2\sqrt{2n \ln n}$ gives a tail at most $1/n$.
    For practical dataset sizes (e.g., $n\ge 10^5$), this deviation is below 1\% of the dataset size.
    Therefore, the dereference table size exhibits strong concentration around its expectation w.h.p.
\end{proof}

\rev{Adding up all the space overhead required for fault tolerance, the space usage of \httwo for $n$ tuples with $w$-bit keys and values is $2nw + 0.22 n\bigl(2w-\log n\bigr)/(1-\delta)$ bits, where $1-\delta$ is the load factor of the dereference table.\mymarginpar{R2D3}
This analysis shows that \httwo provides robust protection against overflow events while maintaining excellent space efficiency.}

%Note that this proof supports deletions and reinsertions because their changes on the Doob martingale are also bounded by two.

\section{Concurrency and Resizing}\label{sec:conc_resize}

\iffalse
\xilin{
\htabbr provides convincing practical operational adaptability. We first describe the concurrency control mechanism used in \htone and \httwo, and then describe our scalable resizing framework and duplicate keys handling.
}
\fi

\subsection{Concurrency}\label{subsec:concurrency}

We employ a fine-grained optimistic concurrency control protocol to support highly parallel operations.
The fundamental mechanism relies on version numbers to ensure atomicity and consistency at the quotient-group granularity.
For \httwo (Figure~\ref{fig:flatten}), the version number is embedded in the header of each cache line, while for \htone it is stored alongside the head-array pointers.

Specifically, we employ \defn{sequence locks} (seqlocks), a common lock primitive used for example in the Linux kernel~\cite{linux-seqlock-doc,linux-seqlock-source}.
Each group maintains a monotonically increasing integer version: even indicates no writer, odd indicates a write in progress.
A writer acquires a write lock on the group by atomically incrementing the version (even to odd), performs its updates, and releases by incrementing again (odd to even).
Readers use optimistic validation: read version $v_0$; if $v_0$ is odd, retry; read the data; read version $v_1$; accept the read only if $v_0 = v_1$, otherwise retry.
%The writer's first increment serves as acquire and the second as release, ensuring readers observe a consistent snapshot.
%This fine-grained, group-level versioning minimizes contention, and the retry-based approach for queries avoids blocking, thus imposing minimal overhead.
A similar concurrency control mechanism is applied to the dereference table, where version numbers are associated with individual bins.

The version counter can wrap around after sufficiently many updates, so there is some risk of the ABA problem.
We use a configurable width version counter to balance space usage and safety.
In our experiments we use a 8-bit counter, which suffices to make wraparound negligible in practice for our workloads.

\subsection{Resizing}\label{subsec:resizing}

\mymarginpar{R4W1}
\mymarginpar{R4D3}
\mymarginpar{R4D4}
\rev{We also provide a scalable resizing framework for \htabbr.
However, resizing is an option that can be enabled for certain applications, rather than a fallback when allocation failure occurs.
Our theory and experiments (Section~\ref{subsec:failure-rate} and \ref{subsec:load_factor}) show that failures occur extremely rarely at reasonable load factors.
As in some prior hash-table designs, in the case of a failure, the system will return an error rather than initiating a resize~\cite{libcuckoo,10.1145/3588727}.
}

Our resizing framework avoids stop-the-world stalls through a key strategy: collaborative, in-flight resizing.
Instead of being blocked, threads accessing a partition during its resize detect the operation and join as workers.
This partition-based approach amortizes growth costs by resizing partitions individually when their load factor exceeds a tunable threshold.
When a resize is triggered, an initiating thread coordinates the migration by dividing it into fine-grained strides, which are then collaboratively processed by the coordinator and any joining threads.
This dynamic distribution of work significantly accelerates the process.
The configurable load threshold and growth factor also enable a direct trade-off between space efficiency and operation latency.

%\alex{I think this is the right story. Let's spell out this tradeoff here. As we do so, we should also specify that if the application knows the upper bound for the size of the table, then better space efficiency is possible. We should also compare to prior resizing (straight doubling), since we do substantially better in the worst case.}

\iffalse
Algorithm~\ref{alg:resizing} outlines this logic, and the implementation is rather complicated.
An incoming operation first checks the partition's state, joining an in-progress resize (\texttt{HelpResize}) or initiating a new one (\texttt{StartResize}) if the load threshold is met.
The \texttt{HelpResize} procedure shows how threads become workers, atomically claiming and executing migration strides.
\texttt{StartResize} details the coordinator's role, from allocating a new table to performing the final atomic pointer swap after all work is done.
This ensures the workload is shared and threads are almost never stalled.
\fi

An additional technique to optimize worst-case space efficiency is \defn{staggering}.
Given a resizing factor $c$ and base capacity $b$, we initialize $k$ partitions with different capacities $b \cdot c^{i/k}$ for $i = 0, 1, \ldots, k-1$.
This approach prevents the scenario where nearly all partitions resize simultaneously, which would result in significant worst-case unused space.
With a doubling resizing factor and 0.7 resizing threshold, staggering achieves a worst-case space efficiency of 43.9\%, compared to 35.6\% for native partitioning.

\section{Evaluation}\label{sec:evaluation}

\begin{figure*}[h]
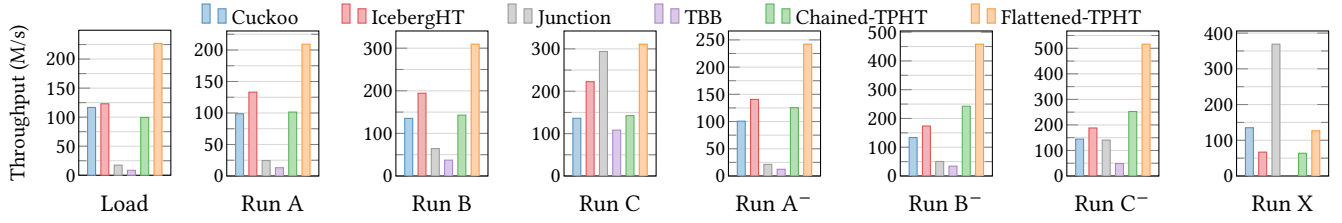

    \centering
    \pgfplotslegendfromname{throughput-legend-horizontal}\vspace{-1em}
    % Single row: Load, Run A, Run B, Run C, Run A^-, Run B^-, Run C^-
    \generateSubfigure{17}{Load}{fill_throughput}{0}%
    \generateSubfigureNoYLabel{17}{Run A}{run_throughput}{0}%
    \generateSubfigureNoYLabel{18}{Run B}{run_throughput}{0}%
    \generateSubfigureNoYLabel{19}{Run C}{run_throughput}{0}%
    \generateSubfigureNoYLabel{20}{Run A$^-$}{run_throughput}{0}%
    \generateSubfigureNoYLabel{21}{Run B$^-$}{run_throughput}{0}%
    \generateSubfigureNoYLabel{22}{Run C$^-$}{run_throughput}{0}%
    \generateSubfigureNoYLabel{28}{Run X}{run_throughput}{0}
    \caption{Performance of hash tables on YCSB workloads with 16 threads. (Throughput is Million ops/second).}
    \label{fig:throughput_subfigures}
\end{figure*}

% Local helper: add all object plots for a given case_id
% #1 = case_id
\newcommand{\AddAllObjectPlotsForCase}[1]{%
    % \htthree (object_id=6)
    \addplot[CuckooLineStyle] table[
            x=space_efficiency,
            y=throughput_millions,
            restrict expr to domain={\thisrow{object_id}}{6:6},
            restrict expr to domain={\thisrow{case_id}}{#1:#1}
        ] {\throughputdata};

    % \htfour (object_id=7)
    \addplot[IcebergLineStyle] table[
            x=space_efficiency,
            y=throughput_millions,
            restrict expr to domain={\thisrow{object_id}}{7:7},
            restrict expr to domain={\thisrow{case_id}}{#1:#1}
        ] {\throughputdata};

    % \htfive (object_id=15)
    \addplot[JunctionLineStyle] table[
            x=space_efficiency,
            y=throughput_millions,
            restrict expr to domain={\thisrow{object_id}}{15:15},
            restrict expr to domain={\thisrow{case_id}}{#1:#1}
        ] {\throughputdata};

    % \htsix (object_id=24)
    \addplot[TBBLineStyle] table[
            x=space_efficiency,
            y=throughput_millions,
            restrict expr to domain={\thisrow{object_id}}{24:24},
            restrict expr to domain={\thisrow{case_id}}{#1:#1}
        ] {\throughputdata};

    % \htone (object_id=17)
    \addplot[TPHTLineStyle] table[
            x=space_efficiency,
            y=throughput_millions,
            restrict expr to domain={\thisrow{object_id}}{17:17},
            restrict expr to domain={\thisrow{case_id}}{#1:#1}
        ] {\throughputdata};

    % \httwo (object_id=23)
    \addplot[BlastLineStyle] table[
            x=space_efficiency,
            y=throughput_millions,
            restrict expr to domain={\thisrow{object_id}}{23:23},
            restrict expr to domain={\thisrow{case_id}}{#1:#1}
        ] {\throughputdata};
}

% Common axis style for all space efficiency plots
\pgfplotsset{
    spaceeff axis/.style={
            width=6cm,
            height=3.4cm,
            xlabel={Space Efficiency},
            xmin=0,
            ymin=0,
            grid=major,
            grid style={gray!30},
            tick label style={font=\small},
            label style={font=\small},
            ylabel style={xshift=-5pt},
            scaled ticks=true,
            tick label style={/pgf/number format/fixed,/pgf/number format/precision=2},
            % ytick distance=10,
            minor tick num=1
        },
    spaceeff axis with ylabel/.style={
            spaceeff axis,
            ylabel={Throughput (M/s)}
        }
}

% Helper: create a subplot for a given case
% #1 = case_id, #2 = position (e.g., [b]{0.45\textwidth}), #3 = extra options for legend
\newcommand{\CreateSpaceEffSubplot}[3]{%
    \begin{subfigure}#2
        \centering
        \begin{tikzpicture}
            \begin{axis}[spaceeff axis, #3]
                \AddAllObjectPlotsForCase{#1}
            \end{axis}
        \end{tikzpicture}
        \caption{Case #1}
    \end{subfigure}%
}

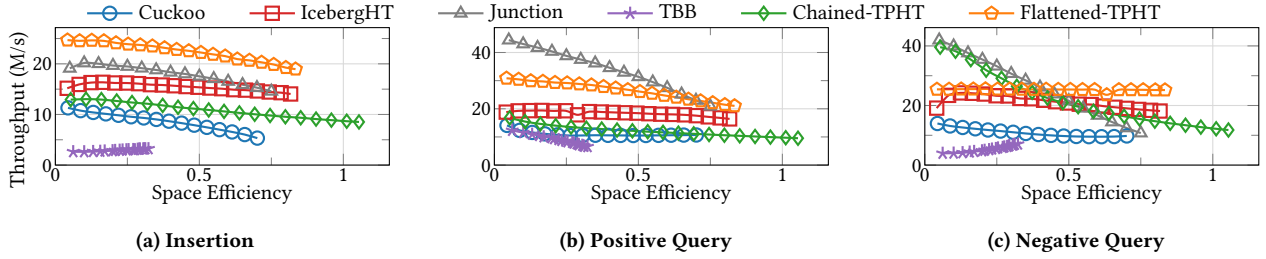
\begin{figure*}[h]
    \centering
    % Legend above figures in 1x5 layout
    {\pgfplotslegendfromname{spaceeff-legend}\vspace{-0.5em}}

    % \vspace{0.5cm}

    % Case 1 with legend definition and y-label
    \begin{subfigure}[b]{0.35\textwidth}
        \centering
        \begin{tikzpicture}
        \begin{axis}[spaceeff axis with ylabel, legend to name={spaceeff-legend}, legend entries={\htthree, \htfour, \htfive, \htsix, \htone, \httwo}, legend cell align=left, legend style={draw=none, font=\small, legend columns=6, /tikz/every even column/.append style={column sep=0.5cm}}]
                \AddAllObjectPlotsForCase{1}
            \end{axis}
        \end{tikzpicture}
        \caption{Insertion}\label{fig:throughput_space_efficiency/insertion}
    \end{subfigure}%
    % Case 9 (no y-label)
    \begin{subfigure}[b]{0.32\textwidth}
        \centering
        \begin{tikzpicture}
            \begin{axis}[spaceeff axis]
                \AddAllObjectPlotsForCase{9}
            \end{axis}
        \end{tikzpicture}
        \caption{Positive Query}\label{fig:throughput_space_efficiency/positive_query}
    \end{subfigure}%
    % Case 10 (no y-label)
    \begin{subfigure}[b]{0.32\textwidth}
        \centering
        \begin{tikzpicture}
            \begin{axis}[spaceeff axis]
                \AddAllObjectPlotsForCase{10}
            \end{axis}
        \end{tikzpicture}
        \caption{Negative Query}\label{fig:throughput_space_efficiency/negative_query}
    \end{subfigure}

    \caption{Throughput-space efficiency tradeoff across insertion and query workloads. Each curve shows how throughput and space efficiency vary with load factor, with rightmost points indicating maximum achievable space efficiency.}\label{fig:throughput_space_efficiency_cases}
\end{figure*}

    We evaluate the performance of tiny pointer hash tables, \htone and \httwo, against state-of-the-art \rev{general} hash tables, {\htthree}~\cite{libcuckoo}, {\htfour}~\cite{10.1145/3588727}, {\htfive}~\cite{junction}, and {\htsix}~\cite{tbb}, \rev{as well as state-of-the-art compact hash tables, Compact Bucketing~\cite{10.1007/s00453-022-00996-y}, Cleary~\cite{10.1109/TC.1984.1676499}, and Layered~\cite{poyias2017mbonsaipracticalcompactdynamic}}.
    \textsf{libcuckoo} implements classical cuckoo hashing with high-performance concurrency control.
    \htfour is a cache-aware hash table that leverages the recent iceberg hashing lemma for improved performance and memory efficiency.
    \htfive provides an efficient implementation of canonical linear probing.
    \htsix is a production separate-chaining hash table from Intel TBB.
    We chose these to cover the major families: separate chaining (\htsix), cuckoo (\htthree), linear probing (\htfive), and cache-aware design (\htfour).
    \rev{For compact hash tables, \htb and \htg are bucketing variants, \htcp uses bidirectional linear probing, \htlp builds layered design upon linear probing, and \htcs \htls are sparse variants.}

Our evaluation addresses the following key questions:
\begin{enumerate}
    \item Do \htabbr{}s deliver leading throughput against state-of-the-art baselines? (Section~\ref{subsec:throughput})
    \item Where do they sit on the speed-space Pareto frontier: how much memory is saved at equal throughput, and what throughput is retained at fixed space efficiency? (Section~\ref{subsec:speed_space_tradeoff})
    \rev{\item How does the performance and space usage of \htone compare to prior compact hash tables? (Section~\ref{subsec:compact_ht_comparison})}
    \item How robust is their performance under scale---with increasing dataset size and thread count? (Section~\ref{subsec:scalability})
    \item How effective is our resizing scheme? Does it avoid blocking the world during resizing? (Section~\ref{subsec:expr_resizing})
    \item Do the core primitives meet their theoretical promises: high dereference-table load factors without failures and near-Poisson occupancy from the hash family under diverse key distributions? (Section~\ref{subsec:load_factor}, Section~\ref{subsec:hash_function_analysis})
\end{enumerate}

\subsection{Experimental Setup}\label{subsec:experimental_setup}

\paragraph{Environment}
All experiments were conducted on a single-socket x86\_64 server running Linux kernel 5.15.0--151-generic.
The machine is equipped with an Intel Xeon Silver 4314 CPU (16 cores, 32 threads; 2.40 GHz base, up to 3.40 GHz) and 128~GiB of RAM, with a 24~MiB shared L3 cache and a single NUMA node.

\paragraph{Datasets}
Following previous work~\cite{10.1145/3588727,10.1145/3725424}, we use the hash-table variant of YCSB~\cite{10.1145/1807128.1807152,indexmicrobench} to evaluate the real-world performance of the proposed hash tables.
YCSB proceeds in two phases: a \defn{load phase} consisting of insertions and a \defn{run phase} consisting of a mix of positive reads and insertions.
Following prior work, we use run phases A, B, and C, which consist of 50\% reads/50\% updates, 90\% reads/10\% updates and 100\% reads, respectively.
We also define phases $\text{A}^-$, $\text{B}^-$, and $\text{C}^-$, which mirror $\text{A}$, $\text{B}$, and $\text{C}$ but use negative rather than positive queries, and phase $\text{X}$, which performs deletions.
In the load phase, we insert 64M tuples into the hash table.
In the run phases A, A$^-$, B, and B$^-$, we insert the same number of tuples, and change the number of queries accordingly.
For C, C$^-$, we query 64M tuples, and for X, we delete 64M tuples.
We also run microbenchmarks over a dataset consisting of uniformly distributed keys.

\paragraph{Scope}
\rev{
    We evaluate all hash tables with 64b keys and 64b values.
    Our benchmarks do not include duplicate keys.\footnote{\rev{Duplicate-key support is compatible with \htabbr designs.
    \htone requires no modification to support duplicate keys, since the hash table key is distinct from the identifier used for tiny-pointer dereference.
    \httwo can be extended to support duplicate keys with light structural changes.}}\mymarginpar{R4D2}
}

\subsection{YCSB}\label{subsec:throughput}

To measure overall performance, we run YCSB with 16 threads.
We initialize all tables with similar initial capacity (except \htsix, which allocates lazily) and terminate at 70\% space efficiency (50\% for Run X), defined as the ratio of inserted data to total allocated memory. The throughput is shown in Figure~\ref{fig:throughput_subfigures}.

\httwo has the highest average throughput with \httwoavgfillthroughput~M ops/s for load phase and \httwoavgrunthroughput~M ops/s for run phases, which is \httwofillspeedup~$\times$ and \httworunspeedup~$\times$ of the fastest baseline, respectively. \htone maintains competitive performance: \htoneavgfillthroughput~M ops/s load phase and \htoneavgrunthroughput~M ops/s run phase, which are \htonefillspeeduppercent\% and \htonerunspeeduppercent\% of the fastest baseline, respectively.

\htthree maintains balanced performance across all workloads (\htthreeavgfillthroughput~M ops/s load phase, \htthreeavgrunthroughput~M ops/s run phase).
This stems from the cuckoo design: each key corresponds to two locations, and probes at most two cachelines per query, leading to predictable access patterns regardless of workload mix.
\htfour uses a front yard/backyard design from recent theory literature~\cite{10.1145/3625817}.
The fingerprint layer of \htfour brings both pros and cons: it adds an extra level of indirection, but also allows for filtering out non-matching keys early.
It run-phase throughput is \htfouroverhtthreerunthroughputpercent\% higher than \htthree.
\htfive is a linear probing hash table with strong positive query performance, \htfiveruncthroughput~M ops/s throughput on Run C, but its complicated concurrency control mechanism for insertions as well as the global load monitoring mechanism leads to contention and poor performance on insertion and mixed workloads.
\htsix maintains stable but lower throughput compared to other baselines.
This likely stems from its separate-chaining design, which provides less cache locality, and RAII accessors that hold locks for the duration of access.
We note that \htsix lacks support for safe multi-threaded deletions, so there are no results for Run X.

\begin{takeaway}
    \httwo runs the fastest on real workloads, while \htone maintains competitive performance despite prioritizing space efficiency.
\end{takeaway}

\ifrevisionmode
{
    \pgfplotsset{
        every axis/.append style={
                % make all plotted curves/bars blue
                every axis plot/.append style={blue},
                % (optional) also make labels/ticks blue:
                xlabel style={text=blue},
                ylabel style={text=blue},
                ticklabel style={text=blue},
                % every axis legend/.append style={text=blue},
            },
    }
    \pgfplotsset{
    spaceeff compact axis/.style={
            width=6cm,
            height=3.4cm,
            xlabel={Space Efficiency},
            xmin=0,
            ymin=0,
            grid=major,
            grid style={gray!30},
            tick label style={font=\small},
            label style={font=\small},
            ylabel style={xshift=-5pt},
            scaled ticks=true,
            tick label style={/pgf/number format/fixed,/pgf/number format/precision=2},
            minor tick num=1
        },
    spaceeff compact axis with ylabel/.style={
            spaceeff compact axis,
            ylabel={Throughput (M/s)}
        }
}

\providecommand{\PlotCompactObj}[4]{% style, obj_id, case_id, data
    \addplot[#1] table[
            x=space_efficiency,
            y=throughput_millions,
            restrict expr to domain={\thisrow{object_id}}{#2:#2},
            restrict expr to domain={\thisrow{case_id}}{#3:#3}
        ] {#4};
}

% Plot all compact objects for a given case_id; optional flag to exclude obj26.
\providecommand{\AddAllCompactPlotsForCase}[1]{%
    % Always plot these
    \PlotCompactObj{TPHTLineStyle}{4}{#1}{\throughputcompactdata}%
    \PlotCompactObj{BlastLineStyle}{23}{#1}{\throughputcompactdata}%
    \PlotCompactObj{CuckooLineStyle}{26}{#1}{\throughputcompactdata}%
    \PlotCompactObj{IcebergLineStyle}{27}{#1}{\throughputcompactdata}%
    \PlotCompactObj{JunctionLineStyle}{28}{#1}{\throughputcompactdata}%
    \PlotCompactObj{TBBLineStyle}{29}{#1}{\throughputcompactdata}%
    \PlotCompactObj{BackupLineStyleOne}{30}{#1}{\throughputcompactdata}%
    \PlotCompactObj{BackupLineStyleTwo}{31}{#1}{\throughputcompactdata}%
}

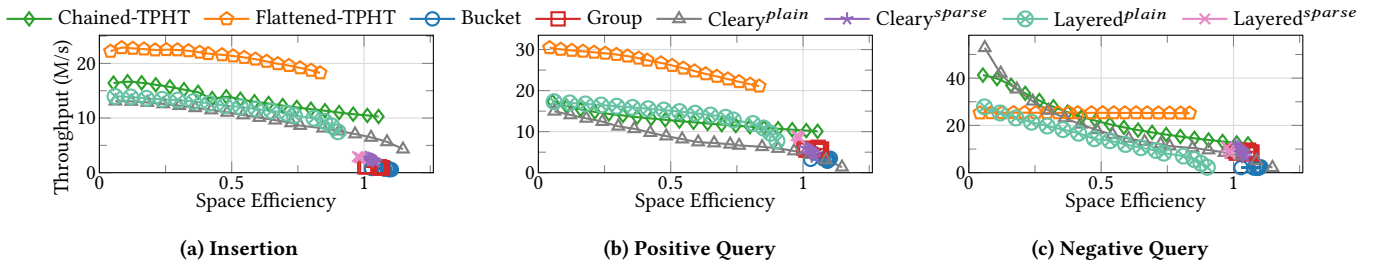
\begin{figure*}[h]
    \centering
    {\pgfplotslegendfromname{spaceeff-legend-compact}\vspace{-0.5em}}

    \begin{subfigure}[b]{0.35\linewidth}
        \centering
        \begin{tikzpicture}
            \begin{axis}[spaceeff compact axis with ylabel, legend to name={spaceeff-legend-compact}, legend entries={\htone, \httwo, \htb, \htg, \htcp, \htcs, \htlp, \htls}, legend cell align=left, legend style={draw=none, font=\small, legend columns=8, /tikz/every even column/.append style={column sep=0.1cm}}]
                \AddAllCompactPlotsForCase{1}
            \end{axis}
        \end{tikzpicture}
        \caption{Insertion}\label{fig:throughput_space_efficiency_compact/insertion}
    \end{subfigure}%
    \begin{subfigure}[b]{0.32\linewidth}
        \centering
        \begin{tikzpicture}
            \begin{axis}[spaceeff compact axis]
                \AddAllCompactPlotsForCase{9}
            \end{axis}
        \end{tikzpicture}
        \caption{Positive Query}\label{fig:throughput_space_efficiency_compact/positive_query}
    \end{subfigure}%
    \begin{subfigure}[b]{0.32\linewidth}
        \centering
        \begin{tikzpicture}
            \begin{axis}[spaceeff compact axis]
                \AddAllCompactPlotsForCase{10}
            \end{axis}
        \end{tikzpicture}
        \caption{Negative Query}\label{fig:throughput_space_efficiency_compact/negative_query}
    \end{subfigure}

    \caption{Throughput-space efficiency tradeoff for compact hash tables.}\label{fig:throughput_space_efficiency_compact_cases}
\end{figure*}

}
\else
    
\fi

% Helper to map case id to label
\newcommand{\getScalingCaseName}[1]{%
  \ifnum#1=1 Insertion\fi%
  \ifnum#1=3 Deletion\fi%
  \ifnum#1=9 Positive Query\fi%
  \ifnum#1=10 Negative Query\fi%
}

% Helper to add all 6 objects with consistent styles
\newcommand{\AddAllObjectsForDataSize}[1]{%
  % \htthree (6)
  \addplot[CuckooLineStyle] table[
      x=table_size,
      y expr=\thisrow{throughput (ops/s)}/1000000,
      restrict expr to domain={\thisrow{case_id}}{#1:#1},
      restrict expr to domain={\thisrow{object_id}}{6:6}
    ] {\datasizedata};
  % \htfour (7)
  \addplot[IcebergLineStyle] table[
      x=table_size,
      y expr=\thisrow{throughput (ops/s)}/1000000,
      restrict expr to domain={\thisrow{case_id}}{#1:#1},
      restrict expr to domain={\thisrow{object_id}}{7:7}
    ] {\datasizedata};
  % \htfive (15)
  \addplot[JunctionLineStyle] table[
      x=table_size,
      y expr=\thisrow{throughput (ops/s)}/1000000,
      restrict expr to domain={\thisrow{case_id}}{#1:#1},
      restrict expr to domain={\thisrow{object_id}}{15:15}
    ] {\datasizedata};
  % \htsix (24)
  \addplot[TBBLineStyle] table[
      x=table_size,
      y expr=\thisrow{throughput (ops/s)}/1000000,
      restrict expr to domain={\thisrow{case_id}}{#1:#1},
      restrict expr to domain={\thisrow{object_id}}{24:24}
    ] {\datasizedata};
  % \htone (17)
  \addplot[TPHTLineStyle] table[
      x=table_size,
      y expr=\thisrow{throughput (ops/s)}/1000000,
      restrict expr to domain={\thisrow{case_id}}{#1:#1},
      restrict expr to domain={\thisrow{object_id}}{17:17}
    ] {\datasizedata};
  % \httwo (20)
  \addplot[BlastLineStyle] table[
      x=table_size,
      y expr=\thisrow{throughput (ops/s)}/1000000,
      restrict expr to domain={\thisrow{case_id}}{#1:#1},
      restrict expr to domain={\thisrow{object_id}}{20:20}
    ] {\datasizedata};
}

% Local axis-label/tick shifts for this figure only
\pgfplotsset{
  data size local shifts/.style={
    xticklabel style={yshift=-1pt},
    yticklabel style={xshift=-1pt},
    xlabel style={yshift=-1pt},
    ylabel style={yshift=4pt, xshift=-5pt}
  }
}

% Axis break zigzag macro (draws on x-axis between positions 0 and 1)
% \newcommand{\DrawAxisBreakZigzag}{%
%   \draw[white, line width=4pt] (axis cs:0.5,0) +(-6pt,-4pt) rectangle +(6pt,4pt);
%   \draw[thick] (axis cs:0.5,0) +(-5pt,-3pt) -- +(-2pt,3pt) -- +(2pt,-3pt) -- +(5pt,3pt);
% }
% % Axis break zigzag macro
\newcommand{\DrawAxisBreakZigzag}{%
  % 1. Reduced eraser width from 8pt to 5pt to match the zigzag exactly
  \fill[white] (axis cs:0.5,0) +(-3pt,-3pt) rectangle +(3pt,3pt);
  % 2. The zigzag (unchanged)
  \draw (axis cs:0.5,0) +(-3pt, 0pt) -- +(-1pt,2pt) -- +(1pt,-2pt) -- +(3pt,0pt);
}

% Shared legend (reuse scaling legend styles)
\begin{figure*}[h]
  \centering
  {\pgfplotslegendfromname{spaceeff-legend}\vspace{-0.8em}}
  % Case 1
  \begin{subfigure}[b]{0.27\textwidth}
    \tiny
    \centering
    \begin{tikzpicture}
      \begin{axis}[
          data size local shifts,
          clip=false,
          width=5cm,
          height=3.2cm,
          xlabel={Dataset Size (MB)},
          ylabel={Throughput (M/s)},
          ymin=0,
          xmin=-0.5,
          xmax=4.5,
          xtick={0,1,2,3,4},
          xticklabels={{1/32},4,32,256,2048},
          x filter/.expression={
            (x < 100000 ? 0 :
            (x < 1000000 ? 1 :
            (x < 10000000 ? 2 :
            (x < 100000000 ? 3 : 4))))
          },
          grid=major,
          grid style={gray!30},
          tick label style={font=\small},
          label style={font=\small},
          scaled ticks=true,
          tick label style={/pgf/number format/fixed,/pgf/number format/precision=1}
        ]
        \AddAllObjectsForDataSize{1}
        \DrawAxisBreakZigzag
      \end{axis}
    \end{tikzpicture}
    \caption{\getScalingCaseName{1}}
  \end{subfigure}%
  % Case 3
  \begin{subfigure}[b]{0.24\textwidth}
    \tiny
    \centering
    \begin{tikzpicture}
      \begin{axis}[
          data size local shifts,
          clip=false,
          width=5cm,
          height=3.2cm,
          xlabel={Dataset Size (MB)},
          ymin=0,
          xmin=-0.5,
          xmax=4.5,
          xtick={0,1,2,3,4},
          xticklabels={{1/32},4,32,256,2048},
          x filter/.expression={
            (x < 100000 ? 0 :
            (x < 1000000 ? 1 :
            (x < 10000000 ? 2 :
            (x < 100000000 ? 3 : 4))))
          },
          grid=major,
          grid style={gray!30},
          tick label style={font=\small},
          label style={font=\small},
          scaled ticks=true,
          tick label style={/pgf/number format/fixed,/pgf/number format/precision=1}
        ]
        \AddAllObjectsForDataSize{3}
        \DrawAxisBreakZigzag
      \end{axis}
    \end{tikzpicture}
    \caption{\getScalingCaseName{3}}
  \end{subfigure}%
  % Case 6
  \begin{subfigure}[b]{0.24\textwidth}
    \tiny
    \centering
    \begin{tikzpicture}
      \begin{axis}[
          data size local shifts,
          clip=false,
          width=5cm,
          height=3.2cm,
          xlabel={Dataset Size (MB)},
          ymin=0,
          xmin=-0.5,
          xmax=4.5,
          xtick={0,1,2,3,4},
          xticklabels={{1/32},4,32,256,2048},
          x filter/.expression={
            (x < 100000 ? 0 :
            (x < 1000000 ? 1 :
            (x < 10000000 ? 2 :
            (x < 100000000 ? 3 : 4))))
          },
          grid=major,
          grid style={gray!30},
          tick label style={font=\small},
          label style={font=\small},
          scaled ticks=true,
          tick label style={/pgf/number format/fixed,/pgf/number format/precision=1}
        ]
        \AddAllObjectsForDataSize{9}
        \DrawAxisBreakZigzag
      \end{axis}
    \end{tikzpicture}
    \caption{\getScalingCaseName{9}}
  \end{subfigure}%
  % Case 7
  \begin{subfigure}[b]{0.24\textwidth}
    \tiny
    \centering
    \begin{tikzpicture}
      \begin{axis}[
          data size local shifts,
          clip=false,
          width=5cm,
          height=3.2cm,
          xlabel={Dataset Size (MB)},
          ymin=0,
          xmin=-0.5,
          xmax=4.5,
          xtick={0,1,2,3,4},
          xticklabels={{1/32},4,32,256,2048},
          x filter/.expression={
            (x < 100000 ? 0 :
            (x < 1000000 ? 1 :
            (x < 10000000 ? 2 :
            (x < 100000000 ? 3 : 4))))
          },
          grid=major,
          grid style={gray!30},
          tick label style={font=\small},
          label style={font=\small},
          scaled ticks=true,
          tick label style={/pgf/number format/fixed,/pgf/number format/precision=1}
        ]
        \AddAllObjectsForDataSize{10}
        \DrawAxisBreakZigzag
      \end{axis}
    \end{tikzpicture}
    \caption{\getScalingCaseName{10}}
  \end{subfigure}%
  \caption{Performance scaling analysis for hash tables with increasing dataset size, single-threaded.}\label{fig:data_size_scaling}
\end{figure*}
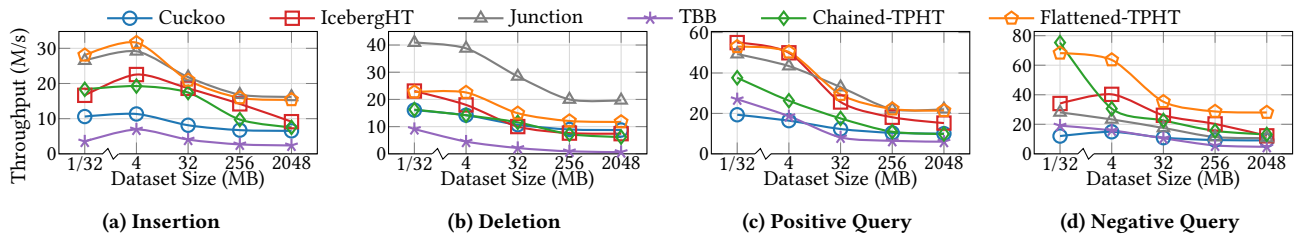

\subsection{Speed/Space Tradeoff}\label{subsec:speed_space_tradeoff}

We define \defn{space efficiency} as the ratio of the data size to maximum memory allocated during execution.
Initializing each hash table's capacity to 16M tuples and inserting data in 5\% increments, we record the memory footprint until reaching 99\% capacity or a failure or resizing occurs.
%For \htthree and \htfive, which use open addressing and resize when insertion fails or a hard-coded load factor is reached, we exclude those points to ensure fair comparison.
%As suggested by the analysis in Section~\ref{subsec:failure-rate}, \htabbr{}s do not fail in these experiments.
We pin to one physical core to ensure stable throughput and memory footprint measurements.
Figures~\ref{fig:throughput_space_efficiency/insertion},~\ref{fig:throughput_space_efficiency/positive_query}, and~\ref{fig:throughput_space_efficiency/negative_query} show the throughput for insertion, positive query, and negative query, respectively.

\htone demonstrates the highest space efficiency of \htonemaxspaceefficiencypercent\%, which shaves memory usage from \htonememshavelowerpercent\% to \htonememshaveupperpercent\% compared with non-tiny-pointer-based baselines.

\httwo achieves the second-highest space efficiency of \httwomaxspaceefficiencypercent\%, has dominating \httwoinsertionspeeddifffivezeropct\% faster insertion performance and \httwonegativequeryspeeddifffivezeropct\% faster negative query performance of the best baseline at 50\% space efficiency and more when space efficiency is above 50\%.
It has the second highest positive query performance at 50\% space efficiency, \httwopositivequerypercentoffastestfivezeropct\% of the best baseline, and becomes the fastest when space efficiency is above 70\%.

\htfive (linear probing) excels at low load: at a 0.05 load factor it achieves \htfiveoverhttwolowloadratio$\times$ \httwo’s positive-query throughput.
Its performance, however, drops sharply---by \htfivethroughputdroppercent\%---as the load factor rises to 0.7.
For comparison, \htone, \httwo, and \htthree drop by \htonethroughputdroppercent\%, \httwothroughputdroppercent\%, and \htthreethroughputdroppercent\%, respectively, whereas \htfour remains flatter with only \htfourthroughputdroppercent\% decline.
This stems from \htfour's small fallback backyard, which becomes increasingly cache-friendly as fallbacks accumulate, partially offsetting other overheads.
With \htsix, insertion and negative-query throughput increases as more data are loaded.
\htsix allocates buckets lazily upon first pointer dereference; loading more data forces allocation of more initial buckets, which raises throughput.

\begin{takeaway}
    \htone achieves above 100\% space efficiency with competitive performance.
    \httwo has better space efficiency than the baselines with competitive-to-faster operations and high throughput even at load factors where the baselines fail.
\end{takeaway}

\subsubsection{Comparison with compact hash tables}\label{subsec:compact_ht_comparison}
\rev{
    In Figure~\ref{fig:throughput_space_efficiency_compact_cases}, we conduct the same experiments as Section~\ref{subsec:speed_space_tradeoff} with compact hash tables.
    Unless a specific operation is indicated, reported throughput is averaged over all data points and operations (insert, query).
    Bucketing methods (\htb and \htg)~\cite{10.1007/s00453-022-00996-y} use frequent bucket reallocations to maintain high space efficiency (up to \compactbucketingmaxspaceefficiencypercent\%).
    However, this results in very slow insertions (\htone is \compactbucketinghtonespeedup$\times$ faster).\mymarginpar{R2D2}
    \htcp~\cite{10.1109/TC.1984.1676499} uses bidirectional linear probing and uses an extra displacement field for each slot to recover the key hash.
    It achieves high space efficiency (\htcpmaxspaceefficiencypercent\%), but still lower throughput than \htone (by \htonebyhtcppercent\%).
    Using \htcs with its sparse layout doesn't result in either high space efficiency or performance.
    \htlp~\cite{poyias2017mbonsaipracticalcompactdynamic} uses a similar linear probing strategy and stores displacement information across multiple levels.
    However, its space efficiency is \htlpmaxspaceefficiencypercent\%, which means its natural comparison point is \httwo rather than \htone.
    However, \httwo has \httwooverhtlpspeedup$\times$ its throughput.
    \htls with sparse layout exchanges \htlsperfdropvshtlppercent\% performance for \htlsspacegainvshtlppercent\% more space efficiency.
}

\subsection{Scalability}\label{subsec:scalability}

\subsubsection{Scaling with dataset size}\label{subsec:dataset_scaling}

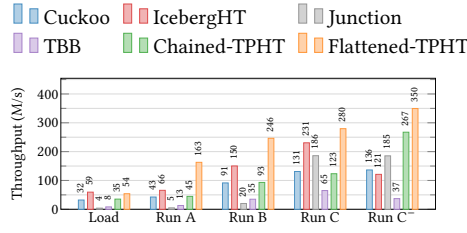
\begin{figure*}[h]
    \centering
    \begingroup
    \makeatletter
    \renewenvironment{figure}[1][]{\begin{minipage}{\linewidth}\centering\def\@captype{figure}}{\end{minipage}}
    \makeatother

    \begin{minipage}[b]{0.32\textwidth}
        \centering
        \pgfplotslegendfromname{throughput-legend-3column}\par
    \end{minipage}\hfill%
    \begin{minipage}[b]{0.66\textwidth}
        \centering
        \pgfplotslegendfromname{resizing-two-legend}\par
    \end{minipage}

    \begin{minipage}[b]{0.32\linewidth}
        \begin{figure}[h]
    \centering
    \resizebox{\linewidth}{!}{%
    \begin{tikzpicture}
        \begin{axis}[
                width=8.5cm,
                height=4cm,
                ylabel={Throughput (M/s)},
                ybar,
                bar width=3pt,
                enlarge x limits=0.15,
                % Use numeric x, but show your symbolic labels:
                xtick={0,1,2,3,4},
                xticklabels={Load,{Run A},{Run B},{Run C},{Run C$^-$}},
                % symbolic x coords={Load, {Run A}, {Run B}, {Run C}, {Run C$^-$}},
                axis lines=box,
                tick align=inside,
                xtick style={draw=none},
                scaled ticks=true,
                tick label style={font=\small, /pgf/number format/fixed,/pgf/number format/precision=1},
                ymajorgrids=true,
                yminorgrids=true,
                minor tick num=1,
                max space between ticks=35pt,
                try min ticks=5,
                grid style={gray!30},
                enlarge y limits={upper,value=0.3},
                ymin=0,
                nodes near coords,
                nodes near coords align={vertical},
                nodes near coords style={/pgf/number format/fixed, /pgf/number format/precision=0, font=\tiny, rotate=90, anchor=west, yshift=0.5pt, xshift=0.3pt},
                legend entries = {\htthree, \htfour, \htfive, \htsix, \htone, \httwo},
                legend cell align = left,
                legend style={draw=none, font=\small, legend columns=3, /tikz/every even column/.append style={column sep=0.1cm}},
                legend to name={throughput-legend-3column},
                unbounded coords=discard,
                filter discard warning=false,
                label style={font=\small}
            ]

            \addResizingPlots;

        \end{axis}
    \end{tikzpicture}%
    }

    \caption{Performance of hash tables on YCSB workloads with 16 threads (resizing enabled).}\label{fig:throughput_subfigures_resizing}
\end{figure}
    \end{minipage}\hfill
    \begin{minipage}[b]{0.32\linewidth}
        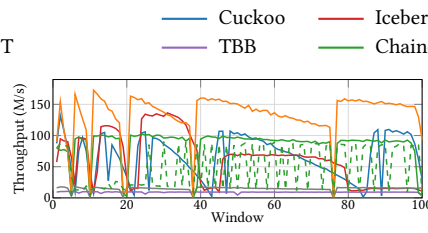
\begin{figure}[h]
    \centering
    \resizebox{\linewidth}{!}{%
    \begin{tikzpicture}
        \begin{axis}[
                width=8.5cm,
                height=3.8cm,
                xlabel={Window},
                ylabel={Throughput (M/s)},
                xmin=0, xmax=100,
                ymin=0,
                grid=both,
                grid style={gray!30},
                legend style={draw=none, font=\small, legend cell align=left, column sep=2pt, /tikz/every even column/.append style={column sep=10pt}},
                legend columns=3,
                legend to name={resizing-two-legend},
                tick label style={font=\small},
                label style={font=\small}
            ]

            \addplot[CuckooLineStyle, mark=none, sharp plot]
            table[x=window, y expr=\thisrow{throughput (ops/s)}/1e6, restrict expr to domain={\thisrow{object_id}}{6:6}] {\progressiveresslidingwindowdata};
            \addlegendentry{\htthree}

            \addplot[IcebergLineStyle, mark=none, sharp plot]
            table[x=window, y expr=\thisrow{throughput (ops/s)}/1e6, restrict expr to domain={\thisrow{object_id}}{7:7}] {\progressiveresslidingwindowdata};
            \addlegendentry{\htfour}

            \addplot[JunctionLineStyle, mark=none, sharp plot]
            table[x=window, y expr=\thisrow{throughput (ops/s)}/1e6, restrict expr to domain={\thisrow{object_id}}{15:15}] {\progressiveresslidingwindowdata};
            \addlegendentry{\htfive}

            \addplot[TBBLineStyle, mark=none, sharp plot]
            table[x=window, y expr=\thisrow{throughput (ops/s)}/1e6, restrict expr to domain={\thisrow{object_id}}{24:24}] {\progressiveresslidingwindowdata};
            \addlegendentry{\htsix}

            \addplot[TPHTLineStyle, mark=none, sharp plot]
            table[x=window, y expr=\thisrow{throughput (ops/s)}/1e6, restrict expr to domain={\thisrow{object_id}}{18:18}] {\progressiveresslidingwindowdata};
            \addlegendentry{\htone}

            \addplot[BlastLineStyle, mark=none, sharp plot]
            table[x=window, y expr=\thisrow{throughput (ops/s)}/1e6, restrict expr to domain={\thisrow{object_id}}{21:21}] {\progressiveresslidingwindowdata};
            \addlegendentry{\httwo}

            \addplot[TPHTLineStyle, mark=none, sharp plot, dashed]
            table[x=window, y expr=\thisrow{throughput (ops/s)}/1e6, restrict expr to domain={\thisrow{object_id}}{25:25}] {\progressiveresslidingwindowdata};

        \end{axis}
    \end{tikzpicture}%
    }
    \caption{Sliding window insertion throughput across windows. The dashed line indicates staggering enabled.}\label{fig:progressive_resizing_sliding_window_throughput}
\end{figure}
    \end{minipage}\hfill
    \begin{minipage}[b]{0.32\linewidth}
        % Resizing RSS plot: reads csv/resizing_rss.csv and plots completion vs memory
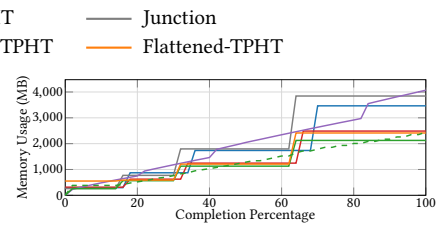
\begin{figure}[h]
    \centering
    \resizebox{\linewidth}{!}{%
    \begin{tikzpicture}
        \begin{axis}[
                width=8.5cm,
                height=3.8cm,
                xlabel={Completion Percentage},
                ylabel={Memory Usage (MB)},
                xmin=0, xmax=100,
                ymin=0,
                grid=both,
                grid style={gray!30},
                tick label style={font=\small},
                label style={font=\small}
            ]

            % List of object ids to plot; adjust if needed
            % Each line filters rows by object_id and plots as a line

            % Object 6 = Cuckoo (htthree)
            \addplot[CuckooLineStyle, mark=none, sharp plot]
            table[x=completion, y=memory_mb, restrict expr to domain={\thisrow{object_id}}{6:6}] {\resizingrssdata};
            \addlegendentry{\htthree}

            % Object 7 = IcebergHT (htfour)
            \addplot[IcebergLineStyle, mark=none, sharp plot]
            table[x=completion, y=memory_mb, restrict expr to domain={\thisrow{object_id}}{7:7}] {\resizingrssdata};
            \addlegendentry{\htfour}

            % Object 15 = Junction (htfive)
            \addplot[JunctionLineStyle, mark=none, sharp plot]
            table[x=completion, y=memory_mb, restrict expr to domain={\thisrow{object_id}}{15:15}] {\resizingrssdata};
            \addlegendentry{\htfive}

            % Object 24 = TBB (htsix)
            \addplot[TBBLineStyle, mark=none, sharp plot]
            table[x=completion, y=memory_mb, restrict expr to domain={\thisrow{object_id}}{24:24}] {\resizingrssdata};
            \addlegendentry{\htsix}

            % Object 18 = Chained-TPHT (htone)
            \addplot[TPHTLineStyle, mark=none, sharp plot]
            table[x=completion, y=memory_mb, restrict expr to domain={\thisrow{object_id}}{18:18}] {\resizingrssdata};
            \addlegendentry{\htone}

            % Object 21 = Flattened-TPHT (httwo)
            \addplot[BlastLineStyle, mark=none, sharp plot]
            table[x=completion, y=memory_mb, restrict expr to domain={\thisrow{object_id}}{21:21}] {\resizingrssdata};
            \addlegendentry{\httwo}

            % Object 25 = Stagger ByteArray Chained HT (htseven)
            \addplot[TPHTLineStyle, mark=none, sharp plot, dashed]
            table[x=completion, y=memory_mb, restrict expr to domain={\thisrow{object_id}}{25:25}] {\resizingrssdata};
            
            \legend{}
            
        \end{axis}
    \end{tikzpicture}%
    }
    \caption{Memory usage during resizing. The dashed line indicates staggering enabled.}\label{fig:resizing_rss_plot}
\end{figure}
    \end{minipage}
    \endgroup
\end{figure*}

\rev{Figure~\ref{fig:data_size_scaling} shows single-threaded throughput
    across dataset sizes from $2^{11}$ to $2^{27}$ on a uniformly distributed
    microbenchmark.
    When the dataset fits entirely in L1 cache ($2^{11}$ entries), \httwo and \htone achieve average throughputs of \httwoscalingsmallavgthroughput~and \htonescalingsmallavgthroughput~M~ops/s across all operations---\httwoscalingsmallavgspeedupoverbest$\times$ and \htonescalingsmallavgspeedupoverbest$\times$ the best baseline, respectively. In particular, \htone dominates negative queries at \htonesmallnegquerythroughput~M~ops/s, as its compact head array keeps lookups L1-resident.
    As the dataset grows beyond cache capacity, throughput decreases across all
    systems; \httwo sustains its lead even at the largest scale, averaging
    \httwoscalinglargeavgthroughput~M/s at 2\,GB---\httwoscalinglargeavgspeedupoverbest$\times$
    the best baseline.
    However, \htfive's deletion outperforms all other systems here because
    a linear probing table treats deletions similarly to positive queries, and since this workload is a single deletion pass, the tombstone mechanism is not stressed.}\mymarginpar{R1D7}

\begin{takeaway}
    \htone and \httwo maintain consistent and expected behavior across different dataset sizes.
    The tiny pointer head array of \htone enables better cache utilization for small table sizes.
\end{takeaway}

\subsubsection{Scaling with threads}\label{subsec:thread_scaling}

We evaluate scalability with 1--16 threads on a microbenchmark with 64M keys.
Most hash tables exhibit nearly linear scaling behavior across all four operation types.
Specifically, the scaling ratios for \htone and \httwo are \htonethreadscalingfactor{} and \httwothreadscalingfactor, respectively.
%\httwo scales effectively due to group-level versioning, which restricts contention to a single cache line.
% Junction's insertion does not exhibit full scalability, which is attributed to a contention during insertion operations.

\begin{takeaway}
    Both tiny pointer hash tables scale near linearly with increasing thread count.
\end{takeaway}

\subsection{Resizing}\label{subsec:expr_resizing}

\subsubsection{Throughput}\label{subsec:resizing_throughput}

Figure~\ref{fig:throughput_subfigures_resizing} analyzes the efficacy of our resizing scheme from Section~\ref{subsec:resizing}.
Unlike Section~\ref{subsec:throughput}, each hash table resizes once during the load phase and once during the run phase (when insertion operations are present).
\httwo achieves \httworesizingavgrunthroughput M ops/s average throughput, ranking the best among all methods.
However, because the insertion of {\htabbr}s carries all the resizing protocol, the load phase throughput decreases by \httworesizingfillthroughputdecreasepercent\% compared to non-resizing tests (Section~\ref{subsec:throughput}).
On the contrary, \htfour exhibits the smallest performance degradation, with a \htfourresizingthroughputdecreasepercent\% throughput decrease due to its proprietary in-place resizing optimization.

\subsubsection{\rev{Progressive resizing throughput}}\label{subsec:resizing_progressive_throughput}

\rev{
    Figure~\ref{fig:progressive_resizing_sliding_window_throughput} shows the throughput on 100 consecutive tumbling windows over multiple resizings.
    \mymarginpar{R4D5}
    % For \httwo, the average throughput drop after resizing is \httworesizingthroughputdroppercent\%, and it recovers quickly.
    For \httwo, throughput drops by \httworesizingthroughputdroppercent\% from immediately after a resize to just before the next resize (after further inserts).
    \htone's throughput tends to be more stable, consistent with evaluation in Section~\ref{subsec:speed_space_tradeoff}.
    \htfour's low-performance period is prolonged as dataset grows, which we attribute to its progressive resizing strategy.
    Moreover, despite the appearance in the figure, staggered resizing in \htone affects overall throughput by less than 10\%.
}

\rev{
    Meanwhile, with resizing enabled under YCSB Run A, the 99th-percentile insertion latency of \httwo and \htone is 541~ns and 778~ns, respectively---nearly an order of magnitude below competitors (e.g., \htthree at 2.7~$\mu$s).
}

\subsubsection{Memory footprint}\label{subsec:resizing_memory_footprint}

We initialize the hash tables with $2^{24}$ capacity and set the double-sizing resizing threshold to 0.7 (except for \htthree, which resizes when insertion fails).
Then we continuously insert $0.6\times 2^{27}$ tuples and show the memory footprint (\texttt{VmHWM}) of the hash tables during the whole resizing process in Figure~\ref{fig:resizing_rss_plot}.
Attributed to the staggering technique in Section~\ref{subsec:resizing}, the footprint of \htone grows steadily instead of stumbling like other baselines.
On the other hand, although \httwo doesn't use staggering, it still has a small memory footprint because of its compact design.
\htone's worst-case space efficiency with staggering is \htsevenworstresizingspaceefficiencypercent\%, while \htfour with the best worst-case space efficiency among baselines is only \htfourworstresizingspaceefficiencypercent\%.

\iffalse
\subsubsection{Percentile latency}\label{subsec:resizing_percentile_latency}

Table~\ref{tab:latency_percentiles} presents percentile latencies for insertion and positive query operations during YCSB Run A execution.
The 99\% percentile latency of \httwo and \htone for both operations is order-of-magnitude lower than competing methods.
For maximum tail latency in positive queries, all baselines are comparable except \htthree, as it's the only one with query-blocking resize.
\htabbr's performance stems from the asynchronous concurrency design of query operations.
\fi

\begin{takeaway}
    The resizing policy simultaneously achieves the goals of high throughput, good space efficiency, and low tail latency.
\end{takeaway}

\subsection{Dereference Table Capacity}\label{subsec:load_factor}

\mymarginpar{R4D3}
\mymarginpar{R4D4}
\mymarginpar{R4W1}
% \xilinrewrite{we configure the dereference table constructed in Section~\ref{subsec:tpimpl} with different bin sizes and show its maximal load factor before failure.}
\rev{
    In order to show that the dereference table implementation (Section~\ref{subsec:tpimpl}) can support high load factors under a sustained workload, we instantiate dereference tables with different bin size and perform sequences of operations.
    We either perform random insertions until allocation failure or alternating random insertions and deletions ($100\times$ the table capacity) at a fixed load factor.
    The results are shown in Figure~\ref{fig:prob_load_mini}.
}
%\xilin{In Fig~\ref{fig:prob_load_mini}, we configure the dereference table (Section~\ref{subsec:tpimpl}) with varying bin sizes.
%Crucially, we want to show it won't uncontrollably trigger failure and fallback (resizing) under reasonable load factor. To do so, we report a table-level load threshold, not a per-allocation failure probability: failures become likely only near the extreme-load point, while under slightly lower threshold we observe sustained failure-free behavior. Recall that the table uses fixed-size bins and two-choice allocation; a failure occurs only when both candidate bins are full.}
%~As discussed in Section~\ref{sec:htone}, the dereference table is the core structure holding the data, and we are able to achieve the expected space efficiency as long as it works well.
% \xilinrewrite{In this experiment, we run two workloads: insertion-only or alternating insertions and deletions. We increase the load factor until an insertion fails.}
%\xilin{We use two stress tests: (i) insertion-only, where load increases until the first failure, and (ii) alternating insert/delete updates, where we first fill to a target load and then run random updates for $100\times$ table capacity while maintaining that load.}
%
As expected, the supported load factor increases with larger bin sizes.
When the bin size reaches $2^7-1$, the dereference table supports 8-bit tiny pointers at load factors \insertiononlyloadfactorpercent\% and \deletionincludedloadfactorpercent\% for insertion-only and alternating workloads, respectively.
%These high load factors are sufficient to achieve superior space efficiency for the entire structure.
%The results confirm previous theoretical analysis that $\log{\delta^{-1}}$-bit tiny pointers can support $1-\delta$ load factors.

\begin{takeaway}
    Our dereference table supports high load factors.
\end{takeaway}

\subsection{Hash Function Analysis}\label{subsec:hash_function_analysis}

To evaluate the quality of the hash function family designed in Section~\ref{subsec:hashfunc}, we examine the relative frequency of bin occupancy.
We use random keys, sequential keys, low Hamming weight keys, or high Hamming weight keys, and hash 64\,M keys to 64\,M bins 100 times.
The results are shown in Fig~\ref{fig:prob_occ_mini}.

Compare against the Poisson distribution ($\lambda = 1$),
the random, low and high Hamming weight keys all closely approximate the Poisson distribution, as expected.
%Both low and high Hamming weight keys achieve similar occupancy distributions to random keys, showcasing the resistance to bit-level dependency.
%The bins that are empty and those holding only one key take \emptybinspercent\% and \onekeybinspercent\% of the total occupancy, respectively, complying with the theory.
%Their whiskers are at most \maxwhiskerdiffpercent\% different from the mean.
%Only when the occupancy reaches 9 does the relative frequency start to deviate from the Poisson approximation, and the frequency is already negligible below $10^{-7}$.
Sequential keys on the other hand are completely evenly distributed, which is also as expected.
It demonstrates lack of independence, and the negative association of this dependence (see Section~\ref{subsec:hashfunc}).
Thus though the hash function is not independent, it is only dependent in that keys have better than random spread.

\begin{takeaway}
    The hash function family designed in Section~\ref{subsec:hashfunc} shows excellent key-distribution properties.
\end{takeaway}

\section{Related Work}

Research on hash table design has produced a wide range of efficient and specialized implementations.
There has been recent work on open-addressing hash tables, such as Horton Tables~\cite{196284}, Concurrent Robin Hood Hashing~\cite{kelly2018concurrentrobinhoodhashing}, GrowTable~\cite{10.1145/3309206}, and DySECT~\cite{DySECT}.
Zombie Hashing~\cite{10.1145/3725424}, following new theory~\cite{Bender22Linear}, inserts tombstones to improve primary clustering.
DHASH implements non-blocking concurrency~\cite{2022ITPDS..33.3274W}.
VIP hashing adapts online to skewed key popularity~\cite{vip_hashing}.

More work has been done on hardware-aware optimizations for hash tables, such as leveraging SIMD instructions~\cite{swisstables,folly}, FPGAs~\cite{10.1145/3132747.3132756,10.1145/2629582}, GPUs~\cite{10.14778/2809974.2809984,10.1145/3514221.3517911,rabl2023vectorized}, PMEM and external memory~\cite{10.14778/3446095.3446101,li2024read,10.1145/3471485.3471506, 10.14778/3551793.3551839,10.1145/3514221.3517884,Kurpicz2022PaCHashPA,288625}, cache-line awareness~\cite{10.1145/2694344.2694359}, and prefetching techniques~\cite{dramhit,katsarakis2024dlht}.

Besides recent theoretical advances on hash table design~\cite{10.1145/3519935.3519969,farachcolton2025optimalboundsopenaddressing}, adjacent work has also been done on perfect hashing~\cite{10.1145/828.1884,10.5555/2394893.2394911}, filtering~\cite{10.1145/3448016.3452841,10.14778/3523210.3523211}, succinct data structures~\cite{10.1145/1216370.1216372,10.5555/1759210.1759247}, history-independence~\cite{attiya2025historyindependent}, Xarray~\cite{10.1145/3627703.3629588}, balls into bins~\cite{10.5555/646975.711521}, and balanced allocations~\cite{10.1137/1.9781611977073.74,10.1145/3323165.3323203}.

Finally, the broader ecosystem of data structures and systems highlights the versatility of quotient-based and hash-based methods.
These principles underpin a range of high-performance applications such as key-value stores~\cite{10.5555/2616448.2616488,10.1145/2043556.2043558,10.1145/3183713.3196898}, memory caching systems~\cite{memcached,10.5555/2482626.2482662}, and join operators~\cite{10.1145/3677134,10.1145/233269.233337}.

\noindent
\begin{minipage}[t]{0.48\linewidth}
    \centering
    % Left: Load Factor (blue) with shading and markers
    \begin{tikzpicture}
        \begin{axis}[
                width=\linewidth,
                height=3cm,
                xlabel={Bin Size},
                ylabel={\shortstack{Load Factor\\ Support (\%)}},
                xmode=log,
                xmin=2,
                xmax=190,
                ymin=0,
                ymax=105,
                ytick={0,20,40,60,80,100},
                grid=major,
                grid style={gray!30},
                legend style={draw=none, font=\footnotesize, at={(0.5,1.02)}, anchor=south},
                legend columns=1,
                tick label style={font=\small},
                label style={font=\small},
                scaled ticks=true,
                xtick={3,7,15,31,63,127},
                xticklabels={$2^2$,$2^3$,$2^4$,$2^5$,$2^6$,$2^7$}
            ]

            % Legend entries
            \addlegendimage{InsertionOnlyStyle, mark=o, mark size=2.5pt}
            \addlegendentry{Insertion Only}
            \addlegendimage{WithDeletionStyle, mark=square, mark size=2.5pt}
            \addlegendentry{With Deletion}

            % Insertion (CSV) shaded area and line with markers
            \addplot[name path=insLineMini, draw=none] table[
                    x=bin_size,
                    y index=3,
                    col sep=comma,
                    restrict expr to domain={\thisrow{object_id}}{4:4}
                ] {\loadfactordata};
            \addplot[name path=insBaseMini, draw=none] coordinates {(3,0) (127,0)};
            \addplot[fill=exp1!25, fill opacity=0.35, draw=none, forget plot] fill between[of=insLineMini and insBaseMini];
            \addplot[InsertionOnlyStyle, mark=o, mark size=2.5pt] table[
                    x=bin_size,
                    y index=3,
                    col sep=comma,
                    restrict expr to domain={\thisrow{object_id}}{4:4}
                ] {\loadfactordata};

            % Deletion (hardcoded) shaded area and line with markers
            \addplot[name path=delLineMini, draw=none] coordinates {
                    (3,7)
                    (7,43)
                    (15,71)
                    (31,85)
                    (63,91)
                    (127,95)
                };
            \addplot[name path=delBaseMini, draw=none] coordinates {(3,0) (127,0)};
            \addplot[fill=poisson!25, fill opacity=0.35, draw=none, forget plot] fill between[of=delLineMini and delBaseMini];
            \addplot[WithDeletionStyle, mark=square, mark size=2.5pt] coordinates {
                    (3,7)
                    (7,43)
                    (15,71)
                    (31,85)
                    (63,91)
                    (127,95)
                };
        \end{axis}
    \end{tikzpicture}

    \captionof{figure}{Load factor supported by dereference tables (Section~\ref{subsec:tpimpl}).}\label{fig:prob_load_mini}
\end{minipage}%
\hfill%
\begin{minipage}[t]{0.48\linewidth}
    \centering

    % Right: Occupancy (log scale) with boxplots and shaded Poisson
    \begin{tikzpicture}
        \begin{axis}[
                width=\linewidth,
                height=3cm,
                xlabel={Bin Occupancy},
                ylabel={\shortstack{Relative\\ Frequency}},
                ymode=log,
                ymin=1e-9,
                ymax=2.5,
                ytick={1e-1, 1e-3, 1e-5, 1e-7, 1e-9},
                xmin=0,
                xmax=11.5,
                xtick={0,1,3,5,7,9,11},
                xticklabels={0,1,3,5,7,9,11},
                grid=major,
                grid style={gray!30},
                legend style={draw=none, font=\footnotesize, at={(0.4,1.02)}, anchor=south, legend image post style={scale=0.5}},
                legend columns=2,
                legend cell align=left,
                tick label style={font=\small},
                label style={font=\small}
            ]

            % Legend entries
            \addlegendimage{area legend, draw=exp1, fill=exp1!25}
            \addlegendentry{Random}
            \addlegendimage{area legend, draw=exp2, fill=exp2!25}
            \addlegendentry{Sequential}
            \addlegendimage{area legend, draw=exp3, fill=exp3!25}
            \addlegendentry{Low Hamming}
            \addlegendimage{area legend, draw=exp4, fill=exp4!25}
            \addlegendentry{High Hamming}

            % Box plots 0..11 with all key sets
            \foreach \i in {0,1,2,3,4,5,6,7,8,9,10,11} {
                % Random keys (leftmost)
                \pgfplotstablegetelem{\i}{lower_whisker}\of{\occboxrandom}\edef\lw{\pgfplotsretval}
                \pgfplotstablegetelem{\i}{lower}\of{\occboxrandom}\edef\lq{\pgfplotsretval}
                \pgfplotstablegetelem{\i}{median}\of{\occboxrandom}\edef\med{\pgfplotsretval}
                \pgfplotstablegetelem{\i}{upper}\of{\occboxrandom}\edef\uq{\pgfplotsretval}
                \pgfplotstablegetelem{\i}{upper_whisker}\of{\occboxrandom}\edef\uw{\pgfplotsretval}
                \pgfmathsetmacro{\pos}{\i - 0.3}
                \edef\boxopts{boxplot prepared={draw position=\pos, lower whisker=\lw, lower quartile=\lq, median=\med, upper quartile=\uq, upper whisker=\uw}, boxplot/draw direction=y, fill=exp1!25, draw=exp1}
                \expandafter\addplot\expandafter+\expandafter[\boxopts] coordinates {};
                
                % Sequential keys
                \pgfplotstablegetelem{\i}{lower_whisker}\of{\occboxsequential}\edef\lw{\pgfplotsretval}
                \pgfplotstablegetelem{\i}{lower}\of{\occboxsequential}\edef\lq{\pgfplotsretval}
                \pgfplotstablegetelem{\i}{median}\of{\occboxsequential}\edef\med{\pgfplotsretval}
                \pgfplotstablegetelem{\i}{upper}\of{\occboxsequential}\edef\uq{\pgfplotsretval}
                \pgfplotstablegetelem{\i}{upper_whisker}\of{\occboxsequential}\edef\uw{\pgfplotsretval}
                \pgfmathsetmacro{\pos}{\i - 0.1}
                \edef\boxopts{boxplot prepared={draw position=\pos, lower whisker=\lw, lower quartile=\lq, median=\med, upper quartile=\uq, upper whisker=\uw}, boxplot/draw direction=y, fill=exp2!25, draw=exp2, very thick}
                \expandafter\addplot\expandafter+\expandafter[\boxopts] coordinates {};
                
                % Low Hamming keys
                \pgfplotstablegetelem{\i}{lower_whisker}\of{\occboxlowhamming}\edef\lw{\pgfplotsretval}
                \pgfplotstablegetelem{\i}{lower}\of{\occboxlowhamming}\edef\lq{\pgfplotsretval}
                \pgfplotstablegetelem{\i}{median}\of{\occboxlowhamming}\edef\med{\pgfplotsretval}
                \pgfplotstablegetelem{\i}{upper}\of{\occboxlowhamming}\edef\uq{\pgfplotsretval}
                \pgfplotstablegetelem{\i}{upper_whisker}\of{\occboxlowhamming}\edef\uw{\pgfplotsretval}
                \pgfmathsetmacro{\pos}{\i + 0.1}
                \edef\boxopts{boxplot prepared={draw position=\pos, lower whisker=\lw, lower quartile=\lq, median=\med, upper quartile=\uq, upper whisker=\uw}, boxplot/draw direction=y, fill=exp3!25, draw=exp3}
                \expandafter\addplot\expandafter+\expandafter[\boxopts] coordinates {};
                
                % High Hamming keys (rightmost)
                \pgfplotstablegetelem{\i}{lower_whisker}\of{\occboxhighhamming}\edef\lw{\pgfplotsretval}
                \pgfplotstablegetelem{\i}{lower}\of{\occboxhighhamming}\edef\lq{\pgfplotsretval}
                \pgfplotstablegetelem{\i}{median}\of{\occboxhighhamming}\edef\med{\pgfplotsretval}
                \pgfplotstablegetelem{\i}{upper}\of{\occboxhighhamming}\edef\uq{\pgfplotsretval}
                \pgfplotstablegetelem{\i}{upper_whisker}\of{\occboxhighhamming}\edef\uw{\pgfplotsretval}
                \pgfmathsetmacro{\pos}{\i + 0.3}
                \edef\boxopts{boxplot prepared={draw position=\pos, lower whisker=\lw, lower quartile=\lq, median=\med, upper quartile=\uq, upper whisker=\uw}, boxplot/draw direction=y, fill=exp4!25, draw=exp4}
                \expandafter\addplot\expandafter+\expandafter[\boxopts] coordinates {};
            }

            % Shaded Poisson under curve and line on top
            \addplot[fill=poisson!25, fill opacity=0.35, draw=none, forget plot] table[
                    x=x,
                    y=y
                ] {\occpois} |- (axis cs:10,1.1e-9) -- (axis cs:0,1.1e-9) -- cycle;
            \addplot[PoissonTheoryStyle, mark=none] table[
                    x=x,
                    y=y
                ] {\occpois};
            \node[
                anchor=north east,
                font=\tiny,
                fill=white,
                fill opacity=0.8,
                text opacity=1,
                inner sep=1.5pt
            ] at (rel axis cs:0.98,0.98) {\tikz{\draw[color=poisson, thick] (0,0) -- (1,0);}~Poisson($\lambda=1$)};
        \end{axis}
    \end{tikzpicture}

    \captionof{figure}{Bin occupancy compared with smoothed Poisson ($\lambda=1$).}\label{fig:prob_occ_mini}
\end{minipage}
\section{Conclusion}\label{sec:conclusion}

We introduced Tiny Pointer Hash Tables, showing that \emph{tiny pointers} and \emph{key quotienting}---two ideas from theory---can be engineered into production-ready hash tables.
\htone prioritizes memory use and is, to our knowledge, the first practical \rev{succinct} hash table.
\httwo prioritizes latency, placing the common case within a single cache miss and achieving \httwomaxspaceefficiencypercent\% space efficiency with up to \httwospeeduppercent\% higher throughput than strong baselines.
Together, these advances move the latency-space Pareto frontier forward for hash tables.
Beyond raw numbers, TPHT supports deletions, online resizing without global pauses, and clean integration with 64-bit keys/values---making the approach directly usable in systems.
In future work, we see opportunities to adapt the tiny-pointer/dereference-table machinery to other pointer-heavy data structures and to explore adaptive quotienting and layout choices for different workloads.

\begin{acks}
Alex Conway and William Kuszmaul were supported by NSF Award CNS-2504470;
William Kuszmaul was also supported by Jane Street.
\end{acks}

\bibliographystyle{ACM-Reference-Format}
\bibliography{bibliography}

\end{document}